\documentclass[a4paper,12pt]{article}

\usepackage{authblk}
\usepackage{hyperref}
\usepackage{braket}
\usepackage{amssymb}
\usepackage{amsmath}
\usepackage{amsthm}
\usepackage{amscd}
\usepackage[usenames,dvipsnames]{xcolor}
\usepackage{graphicx}
\usepackage{geometry}
\usepackage{mathpazo}

\newtheorem{theorem}{Theorem}[section]

\theoremstyle{definition}

\theoremstyle{remark}
\newtheorem{remark}[theorem]{Remark}

\begin{document}

\title{On thermalization in random free fermions\\
{\large Statistical origins and energy-dependent ETH structure}}

\author[1]{Xiao-Kan Guo \footnote{E-mail: kankuohsiao@whu.edu.cn}}
\affil[1]{\it School of Mathematics \& Physics, Yancheng Institute of Technology, Yancheng 224051, Jiangsu, China}
\author[2]{Zhiqiang Huang \footnote{E-mail: zqhuang@hubu.edu.cn}}
\affil[2]{\it School of Physics, Hubei University, Wuhan 430062, Hubei, China}
\date{\today}

\maketitle

\begin{abstract}
We analyze thermalization and eigenstate thermalization in the random free fermion model, a quadratic Hamiltonian with a Gaussian Orthogonal Ensemble hopping matrix, by an exact random-matrix calculation. For a single-particle observable $A$ we compute the ensemble average and fluctuations of $\operatorname{Tr}(A\rho(t))$: we derive the relaxation envelope $g^2(t/\tau_\lambda)$ with $\tau_\lambda = \hbar/(2\eta\sqrt{N})$ and its oscillatory $|t|^{-3}$ tail, both properties of the semicircular limiting density, and we show that the variance of the fluctuations is $O(1/N)$, so that the relaxation profile self-averages at fixed scaled time. Incorporating particle-number conservation, we prove a factorization theorem for the energy-shell statistics of the many-body eigenstates: the energy-independent ETH structure follows from the statistical independence of the eigenvector matrix and the spectrum, and at fixed realization the energy constraint enters only through a covariance subtraction of relative order $1/N$. A controlled eigenvector--eigenvalue correlation deformation $H(\varepsilon)=H_0+\varepsilon H_1$ then generates, at first order in $\varepsilon$, the energy dependence of the diagonal energy-resolved slope and of the off-diagonal spectral function, a single-particle diagnostic of the ETH structure in the allowed particle-hole sector. Correlation-function fluctuations and the single-particle spectral form factor finally identify the subleading $O(1/N)$ late-time fingerprints of two-level spectral correlations. The model thus separates eigenbasis chaos from spectral chaos and identifies the statistical origin of energy-independent versus energy-dependent ETH.
\end{abstract}
\vspace{\stretch{3}}
\thispagestyle{empty}
\newpage
\pagenumbering{arabic}

\section{Introduction}
\label{sec:intro}
The question of how an isolated quantum system, evolving unitarily, can exhibit thermodynamic behavior is central to quantum statistical mechanics. The eigenstate thermalization hypothesis (ETH) \cite{Deutsch1991, Srednicki1994} offers a compelling answer: The individual energy eigenstates themselves encode thermal properties, so that the long-time average of any few-body observable coincides with its microcanonical expectation. 
While ETH is widely believed to hold for chaotic quantum systems \cite{Srednicki1999},  analytical examples are scarce, and in most models exhibiting ETH one resorts to numerical experiments \cite{D'Alessio2016}. Since quantum chaos can be studied via the random matrix theory \cite{GMW98,Mehta2004,Haake2010}, one naturally expects analytical treatments of ETH relying on random matrices.

Mag\'an \cite{Magan2016} introduced an exactly solvable model of random free fermions governed by a quadratic Hamiltonian with a Gaussian Orthogonal Ensemble (GOE) random hopping matrix. In this model, the eigenstates with a large number of particles satisfy ETH, with correlation functions and entanglement entropies approaching their thermal values up to $O(1/N)$ corrections. This model is particularly attractive because it combines the simplicity of a free system with the statistical properties of random matrices, making it a controlled setting in which the statistical ingredients of thermalization and ETH can be separated and analyzed one by one.

Notice that ETH concerns properties of energy eigenstates, which encode the early-time physics of a chaotic quantum system, whereas the late-time behaviour of chaotic spectra is described by random matrix theory \cite{Haake2010, Mehta2004}. Indeed, higher-order correlation functions such as the out-of-time-ordered correlators (OTOC) characterizing operator growth can be computed in the energy eigenstates \cite{FK19,PFK22}. A possible connection between early-time quantum chaos and late-time random-matrix behaviour is suggested by the Bohigas-Giannoni-Schmit (BGS) conjecture \cite{BGS}, stating that the spectral fluctuations of quantum chaotic systems correspond to random matrix ensembles. The BGS conjecture does not, however, directly justify random-matrix models of ETH: for the BGS correspondence to apply to ETH, further assumptions are needed, except in special cases \cite{MaganWu2024}. There is thus more to study in random-matrix models for thermalization beyond ETH itself.

The so-called BGS-inspired random-matrix thermalization has been investigated recently by Weidenm\"uller \cite{Weidenmuller2024a, Weidenmuller2024b}. Treating the Hamiltonian as a member of the GOE or GUE, he computed the ensemble average and variance of $\operatorname{Tr}(A\rho(t))$, demonstrating that the ensemble average relaxes along an envelope determined solely by the average level density. His analysis reveals the central role of the Fourier transform of the level density, $g(t)$, and shows that thermalization carries an explicit time dependence, in contrast to the time-independent statements of ETH \cite{W25}.

In this paper we study this BGS-inspired random-matrix thermalization in the random free fermion model, and use the model to identify which statistical structure of the random Hamiltonian generates which layer of thermalization and ETH. What is new relative to the single-particle random-matrix treatment of \cite{Weidenmuller2024a,Weidenmuller2024b} is the many-body content: the fixed-particle-number Slater-determinant structure of the model connects the single-particle Haar statistics to the energy-shell statistics of the many-body eigenstates (Sec.~\ref{sec:conservation}), and a controlled deformation switches on the energy dependence that the undeformed ensemble lacks. Our results are organized around two main statements, with the remaining sections providing their dynamical and numerical support.

\emph{Main result I: exact statistical separation in the undeformed ensemble.} We first compute the ensemble average and fluctuations of $\operatorname{Tr}(A\rho(t))$ for a single-particle observable $A$: the leading relaxation envelope $g^2(t/\tau_\lambda)$ with $\tau_\lambda = \hbar/(2\eta\sqrt{N})$, its oscillatory $|t|^{-3}$ tail associated with the semicircular limiting density, and the variance of the fluctuations vanishing as $O(1/N)$. Throughout, eigenvector moments are treated with the exact Haar (Weingarten) calculus on $O(N)$ rather than the Gaussian approximation, which is essential at the $O(1/N)$ level and guarantees that the exact initial condition at $t=0$ is respected. This not only confirms the ETH results of \cite{Magan2016} from a dynamical perspective, but also yields explicit expressions for the relaxation envelope, the characteristic time scale, and the fluctuations.
 
Then, we extend the analysis to incorporate particle number conservation, and derive the effective temperature and chemical potential from the initial particle number and energy density.
 We further compute, with the particle number conservation, \emph{energy-shell} eigenstate statistics and prove a factorization theorem: because the eigenvector matrix and the eigenvalues are independent in this ensemble, the joint law of an eigenstate expectation value and the eigenstate energy factorizes. At the ensemble level, the microcanonical average is exactly energy-independent and the conditional variance equals its unconditional Haar value, while at fixed realization the energy-shell conditioning alters the fluctuations only through a covariance subtraction of relative order $1/N$. 
The extension to $k$-body observables is carried out at full coefficient level for $k=2$; the exact finite-population variance and covariance formulas are derived and verified numerically, showing that the covariance subtraction remains of relative order $1/N$ for every fixed $k$.

 \emph{Main result II: controlled breaking of the separation.} The natural converse question---where the nontrivial energy dependence of the ETH functions comes from---is answered constructively by a controlled correlation deformation $H(\varepsilon)=H_0+\varepsilon H_1$: at first order in $\varepsilon$, the diagonal energy-resolved slope and the off-diagonal spectral function acquire energy dependence proportional to $\varepsilon\operatorname{Tr}(AH_1)$ and $\varepsilon\operatorname{Tr}(A^2H_1)(2E-\omega)$, respectively. Unlike Refs.~\cite{H26,H26b}, which analyze the energy-dependent ETH structure through resolvent hierarchies, here we use the exactly solvable random-free-fermion ensemble to identify the statistical origin of energy-independent and energy-dependent ETH and to connect it directly to dynamical thermalization.
Recently, Vallini {\it et. al.} \cite{Vallini2026} developed a
complementary, ensemble-level refinement of ETH based on the principle of
local rotational invariance. In their construction a \emph{global}
Haar rotation $A\mapsto U^{\dagger}AU$ of the observable yields a completely
random eigenbasis and an energy-independent (featureless) ETH structure,
whereas restricting the random basis transformations to windows of width
$\Delta$ in energy restores the smooth energy dependence of the ETH
functions and reorganizes the $1/D$ subleading corrections to the
factorization in terms of  local free cumulants. The random free fermion
model studied here realizes the \emph{global} case exactly in the sense that its
Haar-distributed eigenvector matrix is statistically independent of the
spectrum, while the controlled deformation $H(\varepsilon)=H_0+\varepsilon
H_1$ is a microscopic mechanism of the second, \emph{local} type, admixing predominantly eigenstates that are close in energy. Our
results may therefore be read as an exactly solvable complementation of their general
free-probability treatment.

As supporting dynamics, we compute the  fluctuations of correlation functions and compare the fluctuation scales with those of a full GOE random-matrix reference without particle-number conservation, revealing quantitative differences rooted in the finite-dimensional Haar statistics of the eigenstates; the single-particle spectral form factor then isolates the subleading $O(1/N)$ late-time fingerprints of two-level spectral correlations.

The central organizing question of this paper is: \emph{which statistical ingredient of the random Hamiltonian is responsible for which physical layer of thermalization and ETH?} We answer it with a four-way decomposition: (i) the Haar statistics of the eigenvectors set the ETH fluctuation scale and the off-diagonal amplitude; (ii) the one-point spectral density fixes the leading relaxation envelope $g^2(t/\tau_\lambda)$; (iii) eigenvector-eigenvalue correlations provide a controlled mechanism for generating the energy dependence of the ETH structure; and (iv) two-level spectral correlations contribute, for the observables and time regimes analyzed here, only the subleading late-time corrections. In the undeformed GOE ensemble only (i) and (ii) are active, which is precisely why the model realizes an exactly energy-independent (weak) ETH, while the controlled deformation $\varepsilon H_1$ switches on (iii) at first order in $\varepsilon$. The main text is organized around this decomposition: Sec.~\ref{sec:thermalization} establishes layers (i) and (ii), Sec.~\ref{sec:conservation} layers (i)--(iii), and Sec.~\ref{sec:dynamics} the dynamical fingerprints of layer (iv).

The paper is organized as follows. Section~\ref{sec:model} defines the random free fermion model and recalls its basic statistical properties. Section~\ref{sec:thermalization} contains the core thermalization analysis, in which we compute the ensemble average of $\operatorname{Tr}(A\rho(t))$, derive the relaxation envelope $g(t/\tau_\lambda)$, and establish that the variance of the fluctuations vanishes as $O(1/N)$. Particle number conservation is treated in Sec.~\ref{sec:conservation}, where we connect the filling fraction $f = N_p/N$ and the energy density to an effective temperature and chemical potential, where we prove the factorization theorem for the energy-shell statistics of the many-body eigenstates together with the extension to $k$-body observables, and where a controlled eigenvector-eigenvalue correlation deformation generates the energy-dependent ETH structure at first order. Section~\ref{sec:dynamics} is devoted to the dynamics of correlation functions and the single-particle spectral form factor. We conclude in Sec.~\ref{sec:conclusion} with a summary and outlook. 

\section{Random free fermions}
\label{sec:model}

We consider the random free fermion Hamiltonian   \cite{Magan2016}
\begin{equation}
H = \alpha \sum_{i=1}^{N} c_i^\dagger c_i + \eta \sum_{i,j=1}^{N} c_i^\dagger V_{ij} c_j,
\label{eq:H}
\end{equation}
where $c_i^\dagger$, $c_i$ are fermionic creation and annihilation operators, and $V=(V_{ij})$ is a real \emph{symmetric} $N\times N$ random matrix drawn from the Gaussian Orthogonal Ensemble (GOE) in Mehta's normalization, i.e. the matrix elements $V_{ij}$ with $i\le j$ are independent real Gaussian variables with zero mean and variance $1+\delta_{ij}$ (unit variance for $i<j$, variance $2$ on the diagonal), and $V_{ji}=V_{ij}$. With this normalization the average level density is the semicircle of radius $R=2\sqrt{N}$,  and the eigenvector matrix is distributed according to the Haar measure on $O(N)$. The parameters $\alpha$ and $\eta$ set the energy scales.

Because the Hamiltonian is quadratic, it can be diagonalized exactly. Let $\epsilon_a$ ($a=1,\ldots,N$) be the eigenvalues of $V$, and let $\psi_i^a$ be the components of the corresponding orthonormal eigenvectors. Then
\begin{equation}
H = \sum_{a=1}^N (\alpha + \eta\epsilon_a) \, d_a^\dagger d_a \equiv \sum_{a=1}^N E_a \, d_a^\dagger d_a,
\label{eq:Hdiag}
\end{equation}
with $d_a = \sum_i \psi_i^a c_i$. The many-body eigenstates are Slater determinants constructed from $N_p$ distinct single-particle orbitals,
\begin{equation}
|\Psi^{N_p}\rangle = \prod_{a\in\mathcal{A}} d_a^\dagger |0\rangle, \qquad \mathcal{A}\subset\{1,\ldots,N\},\; |\mathcal{A}| = N_p.
\end{equation}
The corresponding many-body energy is $E_{\mathcal{A}} = \sum_{a\in\mathcal{A}} E_a$.

The eigenvector matrix $(\psi_i^a)$ is exactly Haar distributed on $O(N)$, and it is statistically independent of the eigenvalues---two separate properties: the former fixes the marginal law of each eigenvector, the latter the joint law of eigenvectors and spectrum. The second moment
\begin{equation}
\langle \psi_i^a \psi_j^b \rangle = \frac{1}{N}\,\delta_{ab}\delta_{ij},
\label{eq:psicov}
\end{equation}
is exact, where $\langle \cdot \rangle$ denotes the ensemble average over $V$. In the large-$N$ limit the entries may be treated, to leading order, as independent Gaussian variables, and higher moments are then obtained by Wick contractions. However, this Gaussian approximation misses corrections of relative order $1/N$, at which we will work below. Wherever $O(1/N)$ accuracy matters we therefore use the exact moments of the Haar measure on $O(N)$ (the exact Weingarten functions), collected in Appendix~\ref{app:contractions}.

The eigenvalues $\epsilon_a$ follow Wigner's semicircle law \cite{Mehta2004}
\begin{equation}
\rho_\epsilon(\epsilon) = \frac{2}{\pi R^2}\sqrt{R^2 - \epsilon^2}, \qquad R = 2\sqrt{N},
\end{equation}
so that the single-particle energies $E_a = \alpha + \eta\epsilon_a$ have the average density
\begin{equation}
\rho(E) = \frac{2}{\pi (2\eta\sqrt{N})^2}\sqrt{(2\eta\sqrt{N})^2 - (E-\alpha)^2}, \qquad |E-\alpha| \le 2\eta\sqrt{N}.
\label{eq:rhoE}
\end{equation}
The semicircle law is the $N\to\infty$ limit of the average one-point density: at finite $N$ the density carries corrections $\delta\rho_N$ that we do not need to evaluate explicitly, and any quantity below computed with $\rho(E)$ is understood at the large-$N$ leading order.

Throughout the paper we distinguish three tiers of precision: (i) \emph{exact finite-$N$} statements, namely the Haar/Weingarten moments, initial conditions, and finite-population identities, which hold for every $N$; (ii) \emph{large-$N$ leading-order} statements obtained from the semicircular limiting density, such as the envelope $g(t/\tau_\lambda)$ and its $t^{-3}$ tail; and (iii) \emph{subleading corrections}, of order $1/N$, $O(1/N^2)$, or first order in the deformation parameter $\varepsilon$. Numerical data serve as verification of the formulas, not as proof.

The first two moments of the energy of an $N_p$-particle eigenstate are \cite{Magan2016}
\begin{align}
\Bigl\langle \sum_{a\in\mathcal{A}} E_a \Bigr\rangle &= N_p\alpha, \\
\Bigl\langle \Bigl(\sum_{a\in\mathcal{A}} E_a\Bigr)^2 \Bigr\rangle - \Bigl\langle \sum_{a\in\mathcal{A}} E_a \Bigr\rangle^2 &= \eta^2\, N_p\,(N - N_p + 2).
\end{align}
The variance follows from the exact identities $\sum_{a}\langle \epsilon_a^2\rangle = N(N+1)$ and $\langle (\operatorname{Tr} V)^2\rangle = 2N$, together with $\langle\epsilon_a\epsilon_b\rangle = -1$ for $a\neq b$; it reduces to $N_p N \eta^2$ at low filling, $N_p \ll N$.
These statistical properties form the basis for all subsequent calculations.

The model is analytically tractable because the eigenvector matrix is Haar distributed and statistically independent of the eigenvalues: all higher-point functions are computable in closed form through the Weingarten calculus, and the semicircular level density enters through elementary Fourier transforms. It is this independence---rather than the semicircle law alone---that makes the model a controlled setting for separating the statistical ingredients of thermalization and ETH.

\section{Thermalization: ensemble average and fluctuations}
\label{sec:thermalization}
We now analyze the thermalization dynamics of the observable $\operatorname{Tr}(A\rho(t))$, where $A$ is a fixed single-particle Hermitian operator, and $\rho(t)$ is the time-dependent density matrix of the system,
\begin{equation}
\rho(t) = e^{-iHt/\hbar} \Pi \, e^{iHt/\hbar},
\end{equation}
with $\Pi$ a fixed statistical operator satisfying $\operatorname{Tr}\Pi = 1$. The entire time dependence and randomness come from the time-evolution operator $U(t) = e^{-iHt/\hbar}$. This section is a \emph{single-particle} random-matrix calculation as $A$, $\Pi$ and $U(t)$ act on the $N$-dimensional one-body space. Throughout this section we restrict to intensive single-particle observables normalized by $\operatorname{Tr}(A)=O(1)$ and $\operatorname{Tr}(A^2)=O(1)$; the case of extensive observables is commented on at the end of Sec.~\ref{sec:variance}. The genuinely many-body statements of the eigenstate thermalization hypothesis, which involve the $N_p$-particle Slater-determinant eigenstates and their energy-resolved statistics, are taken up in Sec.~\ref{sec:conservation}.

\subsection{Time evolution operator and its moments}\label{3.1}
In the single-particle basis, the matrix elements of $U(t)$ are
\begin{equation}
U_{ij}(t) = \sum_{a=1}^N \psi_i^a \psi_j^a e^{-iE_a t/\hbar}.
\label{eq:Udef}
\end{equation}
We first average over the eigenvectors $\psi_i^a$, keeping the eigenvalues $E_a$ fixed.
Using \eqref{eq:psicov}, we obtain the first moment
\begin{equation}
\langle U_{ij}(t) \rangle_\psi = \delta_{ij} f(t), \qquad f(t) \equiv \frac{1}{N}\sum_a e^{-iE_a t/\hbar}.
\label{eq:U1}
\end{equation}
For the second moment, we compute $\langle U_{ij}(t) U_{kl}^*(t) \rangle_\psi$. Instead of the Gaussian (Wick) approximation,  we use the exact fourth moment of the Haar measure on $O(N)$, i.e.\ the exact Weingarten function given in Appendix~\ref{app:contractions}. A short calculation gives
\begin{equation}
\langle U_{ij}(t) U_{kl}^*(t) \rangle_\psi
= \Bigl[ |f(t)|^2 + \frac{2\bigl(|f(t)|^2-1\bigr)}{(N-1)(N+2)} \Bigr] \delta_{ij}\delta_{kl} + \frac{N\bigl(1-|f(t)|^2\bigr)}{(N-1)(N+2)}\bigl( \delta_{ik}\delta_{jl} + \delta_{il}\delta_{jk} \bigr).
\label{eq:U2psi}
\end{equation}
The connected (correlated) part is
\begin{equation}
\langle U_{ij}(t) U_{kl}^*(t) \rangle_{\psi,\mathrm{corr}}
= \frac{2\bigl(|f(t)|^2-1\bigr)}{(N-1)(N+2)}\,\delta_{ij}\delta_{kl}
+ \frac{N\bigl(1-|f(t)|^2\bigr)}{(N-1)(N+2)}\bigl( \delta_{ik}\delta_{jl} + \delta_{il}\delta_{jk} \bigr).
\label{eq:U2corr}
\end{equation}
Unlike the Gaussian approximation, this expression is consistent with the exact initial condition $U(0)=I$: At $t=0$ we have $|f(0)|^2=1$, and \eqref{eq:U2psi} reduces exactly to $\delta_{ij}\delta_{kl}$, as it should be. The connected part is time-dependent and vanishes at $t=0$; only at late times, when $|f(t)|^2\to 0$, does it approach the Gaussian-approximation value $\frac{1}{N}\bigl(\delta_{ik}\delta_{jl}+\delta_{il}\delta_{jk}\bigr)$. The factor $\bigl(1-|f(t)|^2\bigr)$, which the Wick contraction would have missed, is essential for the $O(1/N)$ accuracy used in the later variance analysis.  

Next, we average over the eigenvalues $E_a$. Using the average level density \eqref{eq:rhoE}, we find
\begin{equation}
\langle f(t) \rangle = \int dE\,\rho(E)\, e^{-iEt/\hbar} = e^{-i\alpha t/\hbar} \frac{2}{\pi} \int_{-1}^{1} dx\,\sqrt{1-x^2}\, e^{-i(2\lambda t/\hbar)x},
\end{equation}
where $\lambda \equiv \eta\sqrt{N}$. The integral is the Fourier transform of the semicircle, giving
\begin{equation}
\langle f(t) \rangle = e^{-i\alpha t/\hbar}\, g\!\left(\frac{t}{\tau_\lambda}\right), \qquad g(z) = \frac{2J_1(z)}{z},
\label{eq:favg}
\end{equation}
with $J_1(z)$ the Bessel function and
 the characteristic time scale
\begin{equation}
\tau_\lambda \equiv \frac{\hbar}{2\lambda} = \frac{\hbar}{2\eta\sqrt{N}}.
\label{eq:tau_lambda}
\end{equation}
Equation~\eqref{eq:favg} is a large-$N$ leading-order statement in the sense of Sec.~\ref{sec:model}: in general $\langle f(t)\rangle = e^{-i\alpha t/\hbar}g(t/\tau_\lambda) + \delta f_N(t)$, where $\delta f_N(t)$ is the finite-$N$ correction from the difference between the average one-point density and the semicircle; its explicit form is not needed in what follows.
The average of $|f(t)|^2$ splits into an uncorrelated part and a part coming from eigenvalue correlations:
\begin{equation}
\langle |f(t)|^2 \rangle = |\langle f(t) \rangle|^2 + \langle |f(t)|^2 \rangle_{\mathrm{corr}}.
\end{equation}
The correlated part involves the GOE two‑level correlation function and is of order $1/N$ \cite{Weidenmuller2024a}. Hence, to leading order in $1/N$,
\begin{equation}
\langle |f(t)|^2 \rangle = |g(t/\tau_\lambda)|^2 + O(1/N).
\label{eq:f2avg}
\end{equation}

\subsection{Ensemble average of $\operatorname{Tr}(A\rho(t))$}

We now assemble the results to compute $\langle \operatorname{Tr}(A\rho(t)) \rangle$. Following Weidenm\"uller \cite{Weidenmuller2024a}, we write it as the sum of uncorrelated and correlated parts,
\begin{equation}
\langle \operatorname{Tr}(A\rho(t)) \rangle = \operatorname{Tr}\!\Bigl[ A \langle U(t) \rangle_\psi \Pi \langle U^\dagger(t) \rangle_\psi \Bigr] + \operatorname{Tr}\!\Bigl[ A \langle U(t)\Pi U^\dagger(t) \rangle_{\psi,\mathrm{corr}} \Bigr].
\end{equation}
Using \eqref{eq:U1} and the exact expression \eqref{eq:U2corr}, the first term gives $|f(t)|^2 \operatorname{Tr}(A\Pi)$, while the second term yields $\frac{2(|f(t)|^2-1)\operatorname{Tr}(A\Pi) + N(1-|f(t)|^2)\bigl[\operatorname{Tr}(A) + \operatorname{Tr}(A\Pi)\bigr]}{(N-1)(N+2)}$, since $\operatorname{Tr}(A^{\mathsf{T}}\Pi)=\operatorname{Tr}(A\Pi)$ for Hermitian $A$ and $\Pi$. Thus, after averaging over the eigenvectors, the result is exact and takes the compact form
\begin{align}
\bigl\langle \operatorname{Tr}(A\rho(t)) \bigr\rangle_\psi
&= \gamma \,|f(t)|^2 + \mu\,\bigl(1-|f(t)|^2\bigr),
\label{eq:main_avg_psi}
\end{align}
with the deterministic coefficients
\begin{equation}
\gamma = \operatorname{Tr}(A\Pi),
\qquad
\mu = \frac{N\,\operatorname{Tr}(A) + (N-2)\operatorname{Tr}(A\Pi)}{(N-1)(N+2)}
= \frac{\operatorname{Tr}(A) + \operatorname{Tr}(A\Pi)}{N}\,\bigl[1+O(1/N)\bigr].
\label{eq:gamma_mu}
\end{equation}
Note that at $t=0$ (where $|f(0)|^2=1$) the right‑hand side equals $\operatorname{Tr}(A\Pi)$ \emph{exactly}, as required by $\rho(0)=\Pi$. The Gaussian approximation would have violated this initial condition by an $O(1/N)$ term. 

Averaging over the eigenvalues and using the decomposition \eqref{eq:f2_RMT} in Appendix \ref{PE}, i.e.\ $\langle |f(t)|^2\rangle = |g(t/\tau_\lambda)|^2 + \frac1N\mathcal{F}(t/\tau_d) + \frac1N\bigl(1-|g(t/\tau_\lambda)|^2\bigr)$, we obtain, to order $1/N$,
\begin{align}
\bigl\langle \operatorname{Tr}(A\rho(t)) \bigr\rangle
&= |g(t/\tau_\lambda)|^2 \operatorname{Tr}(A\Pi)
 + \frac{1-|g(t/\tau_\lambda)|^2}{N}\bigl[ \operatorname{Tr}(A) + \operatorname{Tr}(A^{\mathsf{T}}\Pi) + \operatorname{Tr}(A\Pi) \bigr] \notag \\
&\quad + \frac{1}{N}\,\mathcal{F}\!\left(\frac{t}{\tau_d}\right) \operatorname{Tr}(A\Pi) + O(N^{-2}) + O\bigl(N^{-2}\operatorname{Tr}(A)\bigr).
\label{eq:main_avg}
\end{align}
Here $\mathcal{F}(t/\tau_d)$ is the two-level function of order unity defined in \eqref{eq:spectral_correction}--\eqref{eq:F_def}, arising from the eigenvalue two‑level correlations; it is the dynamical signature of spectral chaos. For observables normalized by $\operatorname{Tr}(A)=O(1)$ the error terms are uniformly $O(N^{-2})$.
For an intensive observable with $\operatorname{Tr}(A)=O(1)$, the first term $\operatorname{Tr}(A\Pi)|g|^2$ is the leading dynamics and all remaining terms of \eqref{eq:main_avg}, including the equilibrium value $\frac{1}{N}\operatorname{Tr}(A)$, are uniformly $O(1/N)$. For an extensive observable with $\operatorname{Tr}(A)=O(N)$, the equilibrium value $\frac{1}{N}\operatorname{Tr}(A)$ is $O(1)$ and both terms of the large-$N$ limit must be kept:
\begin{equation}
\bigl\langle \operatorname{Tr}(A\rho(t)) \bigr\rangle \;\longrightarrow\; \frac{1}{N}\operatorname{Tr}(A) + \operatorname{Tr}(A\Pi) \, |g(t/\tau_\lambda)|^2.
\label{eq:limit}
\end{equation}
The constant term $\frac{1}{N}\operatorname{Tr}(A)$ is the infinite‑temperature equilibrium value $\operatorname{Tr}(A\rho_{\mathrm{eq}}(\infty))$, reflecting the orthogonal invariance of the ensemble. The time‑dependent term decays from the initial value $\operatorname{Tr}(A\Pi)$ at $t=0$ to zero as $|t|\to\infty$, 
with the relaxation envelope $g^2(t/\tau_\lambda)$ given by \eqref{eq:favg}. The limit in \eqref{eq:limit} is understood at fixed scaled time $s=t/\tau_\lambda$ and matches the exact initial condition at $t=0$ (where the $O(1/N)$ corrections of \eqref{eq:main_avg} vanish identically); the large-$N$ and long-time limits are not interchangeable, and their relation is discussed in Sec.~\ref{sec:variance}.
For large $|t|$, we use the asymptotic expansion of the Bessel function $J_1(z)$:
\begin{equation}
J_1(z) \sim \sqrt{\frac{2}{\pi z}} \cos\!\left(z - \frac{3\pi}{4}\right), \qquad z\to\infty.
\end{equation}
Hence
\begin{equation}
g(z) \sim \sqrt{\frac{8}{\pi}} \, \frac{\cos(z - 3\pi/4)}{z^{3/2}}, \quad\text{and}\quad g^2(z) \sim \frac{8}{\pi} \frac{\cos^2(z - 3\pi/4)}{z^3}.
\end{equation}
Therefore,
\begin{equation}
g^2(t/\tau_\lambda) = O\!\left( \frac{1}{|t|^3} \right), \qquad |t|\to\infty.
\label{eq:g2asym}
\end{equation}
The oscillatory $|t|^{-3}$ tail follows directly from the square-root edges of the semicircular limiting density \eqref{eq:rhoE} in the large-$N$ regime; it is a property of the semicircular edge, not a model-independent law.

\subsection{Variance and the vanishing of fluctuations}
\label{sec:variance}

To show that the relaxation profile self-averages over realizations, we estimate the variance of $\operatorname{Tr}(A\rho(t))$,
\begin{equation}
\operatorname{Var}\!\bigl( \operatorname{Tr}(A\rho(t)) \bigr) = \bigl\langle \operatorname{Tr}(A\rho(t)) \operatorname{Tr}(A\rho(t)) \bigr\rangle - \bigl\langle \operatorname{Tr}(A\rho(t)) \bigr\rangle^2.
\label{eq:variance}
\end{equation}
We follow the systematic analysis of Weidenm\"uller \cite{Weidenmuller2024a} and classify the contributions according to how the factors $U(t)$ and $U^\dagger(t)$ are contracted when averaging over the eigenvectors $\psi_i^a$. The contractions yield terms proportional to various products of the functions $f(t)$, $\langle f(t) \rangle$, and the connected parts of higher moments. 
 The key observation is that every connected contribution carries at least one explicit factor $1/N$, rendering the variance $O(1/N)$.
 
  In the classification below we use the Gaussian (Wick) contraction rules for the higher cumulants of $U$, which reproduce the exact Haar moments only up to corrections of relative order $1/N$. This approximation affects only the $O(1/N^2)$ and $O(1/N^3)$ pieces (cases (iv) and (v)); the leading channels identified below---the spectral channel \eqref{eq:var_spectral} and the leading eigenvector-fluctuation channel built from the exact connected second moment---are exact as stated, so the parametric statement $\operatorname{Var} = O(1/N)$ is robust. The exact second moment \eqref{eq:U2psi}, which carries the essential $(1-|f|^2)$ structure, is used throughout.

Expanding the product of the two identical traces $\operatorname{Tr}(A U(t)\Pi U^\dagger(t))$, we have four $U$'s and four $U^\dagger$'s. The average over $\psi$ is performed by applying Wick's theorem to all possible pairings of the eight Gaussian variables. We group the resulting contributions into five cases, following Weidenm\"uller \cite{Weidenmuller2024a}:
\begin{itemize}
    \item \textbf{Case (i):} In each trace, the two factors $U$ and $U^\dagger$ are averaged independently, or they are replaced by the correlated part $\langle UU^\dagger\rangle_{\mathrm{corr}}$. The only non‑vanishing contribution to the variance then comes from eigenvalue correlations, i.e.\ from the fluctuations of the eigenvector-averaged value \eqref{eq:main_avg_psi}. Explicitly, the leading piece is
    \begin{equation}
        \Bigl( \bigl\langle |f(t)|^2 |f(t)|^2 \bigr\rangle - \bigl\langle |f(t)|^2 \bigr\rangle^2 \Bigr) \bigl( \operatorname{Tr}(A\Pi) \bigr)^2.
        \label{eq:casei}
    \end{equation}
    The variance of $|f(t)|^2$ involves the GOE two‑, three‑, and four‑point level correlation functions. At fixed scaled time it is $O(1/N^2)$ rather than $O(1/N)$: in the cumulant expansion the disconnected pieces cancel, and the leading term is built from products of two-point functions,
    \begin{equation}
        \frac{2}{N^2}\,\Bigl[ \bigl(1-|g|^2 + \mathcal{F}\bigr)^2 + \bigl| f(2t) - f(t)^2 + \mathcal{F}_2\bigr|^2 + \cdots \Bigr],
    \end{equation}
    where $\mathcal{F}$ is defined in \eqref{eq:spectral_correction} and $\mathcal{F}_2$ is the analogous Fourier transform of the two-level cluster function weighted by $e^{-i(E_a+E_b)t}$; the three‑ and four‑point functions enter only at $O(1/N^3)$ and beyond. The spectral channel \eqref{eq:var_spectral} is therefore $O(1/N^2)\times\bigl(\gamma-\mu\bigr)^2$ at fixed scaled time $s=t/\tau_\lambda$ (this estimate does not extend to Heisenberg-time scales, where $t/\tau_d\sim 1$ and the two-level correlations become of order one), and it decays further at late times. The cross- and double-correlated replacements complete the coefficient $\bigl(\operatorname{Tr}(A\Pi)\bigr)^2$ into $\bigl(\gamma-\mu\bigr)^2$, see \eqref{eq:var_spectral} below.
    
    \item \textbf{Case (ii):} One factor $U$ or $U^\dagger$ in each trace is replaced by its average, and the remaining two factors are taken from the correlated part of the second moment. There are four possibilities. A representative example is
    \begin{align}
        U^\dagger(t) &\to \langle U^\dagger(t)\rangle,\quad U(t) \to \langle U(t)\rangle, \\
        U(t) U^\dagger(t) &\to \nu_1,
    \end{align}
    where $\nu_1 = \langle U(t)U^\dagger(t)\rangle_{\mathrm{corr}}$ is given in \eqref{eq:U2corr}; it carries the factor $\bigl(1-|f(t)|^2\bigr)/N$. This yields a term proportional to
    \begin{equation}
        \frac{1-|f(t)|^2}{N}\, \operatorname{Tr}(A^2 \Pi^2) \, |g(t/\tau_\lambda)|^2.
    \end{equation}
    For an intensive observable normalized by $\operatorname{Tr}(A^2)=O(1)$, we have $\operatorname{Tr}(A^2\Pi^2)\le \operatorname{Tr}(A^2)=O(1)$, so this term is $O(1/N)$ at generic times. It vanishes identically at $t=0$, where the connected part of the exact second moment \eqref{eq:U2corr} is zero; indeed at $t=0$ the whole eigenvector-fluctuation channel vanishes because $U(0)=I$ exactly, and the variance then vanishes exactly, as required by the deterministic initial value $\operatorname{Tr}(A\rho(0))=\operatorname{Tr}(A\Pi)$. 
    
    \item \textbf{Case (iii):} Each trace is replaced by the correlated part of the second moment. That is,
    \begin{equation}
        \langle U(t)U^\dagger(t)\rangle_{\mathrm{corr}} \quad \text{and} \quad \langle U(t)U^\dagger(t)\rangle_{\mathrm{corr}}.
    \end{equation}
    Using \eqref{eq:U2corr}, this gives a contribution proportional to
    \begin{equation}
        \frac{\bigl(1-|f(t)|^2\bigr)^2}{N^2} \bigl[ \operatorname{Tr}(A^2)\operatorname{Tr}(\Pi^2) + \operatorname{Tr}(A)^2 \operatorname{Tr}(\Pi^2) \bigr] \quad \text{or} \quad \frac{\bigl(1-|f(t)|^2\bigr)^2}{N^2} \operatorname{Tr}(A\Pi)^2,
    \end{equation}
    all of which are $O(1/N^2)$ (since $\operatorname{Tr}(A^2)$ and $\operatorname{Tr}(\Pi^2)$ are $O(1)$).

    \item \textbf{Case (iv):} The three‑factor cumulant of $U$ appears, with the remaining factor averaged. This involves products of three $U$'s (or two $U$'s and one $U^\dagger$) and is proportional to $1/N^2$ times traces like $\operatorname{Tr}(\Pi A A^\mathsf{T})$ or $\operatorname{Tr}(A)\operatorname{Tr}(A\Pi \Pi^\mathsf{T})$, which are $O(1)$ or smaller for the normalized observable; hence this case contributes $O(1/N^2)$.

    \item \textbf{Case (v):} The four‑factor cumulant of $U$ appears. This is proportional to $1/N^3$ times traces such as $\operatorname{Tr}(A^2)\operatorname{Tr}(A\Pi)\operatorname{Tr}(A)$ or $\operatorname{Tr}(\Pi^2)[\operatorname{Tr}(A)]^2$, all $O(1/N^3)$.
\end{itemize}

Collecting all contributions, it is most transparent to organize the variance by the law of total variance,
\begin{equation}
\operatorname{Var}\!\bigl( \operatorname{Tr}(A\rho(t)) \bigr)
= \Bigl\langle \operatorname{Var}_\psi\bigl[ \operatorname{Tr}(A\rho(t)) \bigr] \Bigr\rangle_E
+ \operatorname{Var}_E\Bigl[ \bigl\langle \operatorname{Tr}(A\rho(t)) \bigr\rangle_\psi \Bigr],
\label{eq:variance_split}
\end{equation}
where $\langle\cdot\rangle_\psi$ is the average over the eigenvectors and $\langle\cdot\rangle_E$ the average over the eigenvalues. The first term is the eigenvector-fluctuation channel estimated by cases (ii)--(v): its leading piece is built from the exact connected second moment \eqref{eq:U2corr} together with the exact unconjugated pairing \eqref{eq:UUncorr}, carries a factor $\bigl(1-|f(t)|^2\bigr)|f(t)|^2$ times $O(1)$ traces such as $\operatorname{Tr}(A^2\Pi^2)$, and is $O(1/N)\times\operatorname{Tr}(A^2\Pi^2)$ at generic times, while the three- and four-cumulant pieces are $O(1/N^2)$ and $O(1/N^3)$; the whole channel vanishes identically at $t=0$, since $U(0)=I$ exactly. The second term is the spectral channel of case (i) and is evaluated exactly: by \eqref{eq:main_avg_psi} the eigenvector average is the affine function $\gamma|f(t)|^2 + \mu\bigl(1-|f(t)|^2\bigr)$ of $|f(t)|^2$ with deterministic coefficients, so that
\begin{equation}
\operatorname{Var}_E\Bigl[ \bigl\langle \operatorname{Tr}(A\rho(t)) \bigr\rangle_\psi \Bigr]
= \bigl(\gamma-\mu\bigr)^2 \, \operatorname{Var}\!\bigl( |f(t)|^2 \bigr),
\label{eq:var_spectral}
\end{equation}
where $\operatorname{Var}\bigl(|f(t)|^2\bigr) = \bigl\langle |f(t)|^4\bigr\rangle - \bigl\langle |f(t)|^2\bigr\rangle^2$ is $O(1/N^2)$ at fixed scaled time, as shown in case (i). The variance is therefore governed by two channels of different origin and different parametric size. For an initial state with an $O(1)$ overlap of $\Pi$ with the support of $A$---e.g.\ a laboratory-basis state localized where the observable acts, so that $\gamma$ and $\operatorname{Tr}(A^2\Pi^2)$ are of order one---the eigenvector-fluctuation channel is $O(1/N)$ and \emph{dominates}, while the spectral channel is subleading, $O(1/N^2)$; for Haar-typical initial states with $\gamma,\operatorname{Tr}(A^2\Pi^2)=O(1/N)$, both channels are $O(1/N^2)$. In all cases the upper bound $\operatorname{Var}=O(1/N)$ holds for intensive observables, and the coefficient is determined \emph{not} by spectral statistics alone but by the sum of the two channels. Both channels vanish exactly at $t=0$---the eigenvector-fluctuation channel because $U(0)=I$, the spectral channel because $|f(0)|^2=1$---so the exact initial condition is respected. At late times the $O(1/N)$ channel decays---$|f(t)|^2$ approaches $1/N$ and $f(t),\,f(2t)$ approach zero---and the variance approaches the fluctuation level of the diagonal-ensemble value $\sum_a A_{aa}\Pi_{aa}$, which is only $O(1/N^3)$. We therefore state the result as
\begin{equation}
\operatorname{Var}\!\bigl( \operatorname{Tr}(A\rho(t)) \bigr) = O(1/N),
\label{eq:var_final}
\end{equation}
a parametric statement whose robustness does not require the Gaussian estimates of cases (iv) and (v). Hence, for each fixed scaled time $t/\tau_\lambda$ (i.e., $t$ scaling as $\tau_\lambda \sim N^{-1/2}$), the variance vanishes as $N^{-1}$ in the thermodynamic limit, and the fluctuation amplitude is $O(N^{-1/2})$: at any fixed scaled time,
\begin{equation}
\operatorname{Tr}(A\rho(t)) \longrightarrow \frac{1}{N}\operatorname{Tr}(A) + \operatorname{Tr}(A\Pi) \, |g(t/\tau_\lambda)|^2
\end{equation}
in probability (equivalently, in mean square). The relaxation profile is therefore \emph{self-averaging} in the large-$N$ limit, with negligible sample‑to‑sample variations.

The estimates above assumed that $\operatorname{Tr}(A)$ and $\operatorname{Tr}(A^2)$ are of order unity, i.e., the observable is normalized to have a finite thermodynamic expectation value. This is the physically relevant case (e.g., local densities). For extensive observables with $\operatorname{Tr}(A) \sim N$, the variance would be $O(1)$ and the fluctuations would not vanish; however, such observables are not considered in typical thermalization studies because their expectation values diverge in the thermodynamic limit. In the context of ETH, one always considers intensive or local observables.

Throughout this section the term ``thermalization'' is used in the BGS-inspired random-matrix sense of Sec.~\ref{sec:intro}: relaxation of the ensemble average toward the ensemble equilibrium value \eqref{eq:limit}, together with self-averaging of the fluctuations. The physical many-body statement---that fixed-particle-number eigenstates reproduce finite-temperature grand-canonical expectation values of few-body observables---is established in Sec.~\ref{sec:conservation}.

\subsection{Distinguishing eigenbasis chaos from spectral chaos}
\label{sec:eigenbasis_vs_spectral}
The above analysis of the random free fermion model provides a concrete, analytically tractable illustration of the dichotomy between the {\it eigenbasis chaos}, which controls early-time dynamics such as thermalization and the growth of  OTOC, and the {\it spectral chaos}, which governs very late-time phenomena like the ramp and plateau of the spectral form factor \cite{MaganWu2024}. Since the Hamiltonian of this model contains a random-matrix one-body part \eqref{eq:H}, we see that  its eigenstates satisfy ETH \cite{Magan2016}, and 
the eigenvalues also follow Wigner's semicircle law and display the level repulsion characteristic of random matrix theory \cite{GMW98,Mehta2004,Haake2010}.  However,  the dynamical consequences of these two types of chaos appear on very different time scales and with different parametric dependencies. Throughout, by \emph{eigenbasis chaos} we mean the randomness of the eigenvectors in the laboratory basis, characterized here by Haar statistics, and by \emph{spectral chaos} the fine-grained eigenvalue correlations, i.e., GOE level repulsion and the associated two-level cluster function.

The ensemble average of the observable $\operatorname{Tr}(A\rho(t))$ is given by Eq.~\eqref{eq:main_avg}.  To leading order in $1/N$, this simplifies to
\begin{equation}
\bigl\langle \operatorname{Tr}(A\rho(t)) \bigr\rangle = \operatorname{Tr}(A\Pi) \, |g(t/\tau_\lambda)|^2 + \frac{1-|g(t/\tau_\lambda)|^2}{N}\operatorname{Tr}(A) + O(1/N),
\label{eq:leading_thermalization}
\end{equation}
which matches the exact initial value $\operatorname{Tr}(A\Pi)$ at $t=0$---where the $O(1/N)$ correction vanishes identically---and approaches the infinite-temperature equilibrium $\frac{1}{N}\operatorname{Tr}(A)$ at late times. For observables normalized by $\operatorname{Tr}(A)=O(1)$ the second term is itself $O(1/N)$ and the leading form reduces to $\operatorname{Tr}(A\Pi)|g|^2$. The derivation of $g(z)$ in \eqref{eq:favg} involves only the average level density and the Haar statistics of the eigenvectors.
 The connected spectral correlations enter only through the subleading $O(1/N)$ corrections; at fixed scaled time these are subleading, but at sufficiently late times, when the leading envelope has decayed to $O(1/N)$, they become the dominant observable correction (see below).

The $O(1/N)$ terms in Eq.~\eqref{eq:main_avg} contain contributions from two distinct sources.  The first is the eigenbasis contribution $\frac{1-|g|^2}{N}\bigl[\operatorname{Tr}(A)+\operatorname{Tr}(A^{\mathsf{T}}\Pi)\bigr]$, which arises from the connected part of the \emph{exact} second moment \eqref{eq:U2corr}.  The second contribution comes from the correlated part of $\langle |f(t)|^2 \rangle$, i.e. from the eigenvalue two-point correlation function.  As shown in Appendix B of \cite{Weidenmuller2024a}, this correlated part is
\begin{equation}
\langle |f(t)|^2 \rangle_{\mathrm{corr}} = \frac{1}{N} \bigl[ 1 - \widetilde b(t/\tau_d) \bigr] = \frac{1}{N}\bigl(1-|g(t/\tau_\lambda)|^2\bigr) + \frac{1}{N} \, \mathcal{F}\!\left(\frac{t}{\tau_d}\right),
\label{eq:spectral_correction}
\end{equation}
where the smooth piece $\frac{1}{N}(1-|g|^2)$ is the same as in the Poissonian ensemble, and the additional two-level contribution is
\begin{equation}
\mathcal{F}(t/\tau_d) = |g(t/\tau_\lambda)|^2 - \widetilde b(t/\tau_d), \qquad \widetilde b(s) = \int_{-\infty}^{\infty} Y_2(y)\, e^{2\pi i s y}\, dy,
\label{eq:F_def}
\end{equation}
with $\tau_d = \hbar/d$ the Heisenberg time ($d$ the mean level spacing) and $\widetilde b(s)$ the Fourier transform of the GOE two-level cluster function $Y_2(y)$, normalized by $\widetilde b(0)=1$ and decaying on the scale $s\sim 1$ \cite{Mehta2004}.  Both functions are defined in the standard local (bulk) approximation: after unfolding the spectrum with the local mean level spacing, the two-level cluster function becomes translation invariant, and the semicircular edge regions, where the local spacing varies, contribute only at subleading order in $1/N$; the same convention is used in the numerical comparison of Sec.~\ref{sec:SFF_comparison}.  The function $\mathcal{F}(s)$ vanishes as $s\to 0$ and as $s\to\infty$: at $t=0$ the two terms in \eqref{eq:F_def} cancel identically, as required by the exact initial condition $\langle |f(0)|^2\rangle=1$, and at fixed $s>0$ the envelope $|g(t/\tau_\lambda)|^2=|g(s\tau_d/\tau_\lambda)|^2$ vanishes in the large-$N$ limit, so that $\mathcal{F}(s)\to -\widetilde b(s)$.  The two-level correlations therefore \emph{depress} the form factor below its smooth Poissonian value before the Heisenberg time---the correlation hole---and $\mathcal{F}$ is the dynamical signature of spectral chaos: it vanishes if the eigenvalues are uncorrelated (Poisson statistics), and its specific form depends on the universality class of the random matrix ensemble.

Thus, the full time evolution of $\operatorname{Tr}(A\rho(t))$ can be schematically written as
\begin{equation}
\operatorname{Tr}(A\rho(t)) = \underbrace{\frac{1}{N}\operatorname{Tr}(A) + \operatorname{Tr}(A\Pi) \, |g(t/\tau_\lambda)|^2}_{\text{eigenbasis chaos (leading)}} \;+\; \underbrace{\frac{1}{N}\,\mathcal{F}(t/\tau_d)\operatorname{Tr}(A\Pi) + \cdots}_{\text{spectral chaos (subleading)}}.
\end{equation}
The leading envelope dominates up to the crossover time $t_\times \sim \tau_\lambda N^{1/3}$, set by $|g(t/\tau_\lambda)|^2 \sim 1/N$. Because $\tau_d/\tau_\lambda \sim N$, this crossover time is still much smaller than the Heisenberg time for large $N$.  The spectral-correlation correction therefore becomes parametrically visible well before the Heisenberg time $\tau_d \sim N\tau_\lambda$: In the window $t_\times \lesssim t \lesssim \tau_d$ the $O(1/N)$ correction already exceeds the decaying leading envelope.  Beyond this window the correction varies on the Heisenberg scale and eventually saturates, preventing the expectation value from decaying exactly to zero and giving rise to persistent fluctuations of size $1/N$.  These fluctuations are the dynamical footprint of spectral chaos.

It is instructive to compare the above analysis of random free fermions with the analysis of Poissonian ensemble \cite{MaganWu2024} (cf. Appendix \ref{PE}).
The leading term $|\langle f(t)\rangle|^2$ is identical in both ensembles.  Since the relaxation envelope in Eq.~\eqref{eq:main_avg} is proportional to $|\langle f(t)\rangle|^2$, we conclude that the leading relaxation dynamics described by $g^2(t/\tau_\lambda)$ is insensitive to the fine-grained spectral statistics: within this comparison, which keeps the one-point density and the eigenvector statistics fixed and changes only the eigenvalue correlations, the leading envelope is identical in both ensembles.  This insensitivity is exactly what one expects from the distinction between eigenbasis chaos and spectral chaos: the leading relaxation is controlled by the eigenvector statistics together with the one-point spectral density, while the fine-grained spectral statistics affect only the subleading late-time corrections.

Although the leading relaxation envelope is the same, the fluctuations (variance) of $\operatorname{Tr}(A\rho(t))$ differ between the two ensembles at $O(1/N)$.  In our  model, the variance contains contributions from both the disconnected part (smooth) and the connected part (the two-level hole structure of $\mathcal{F}$).  In the Poissonian ensemble, only the smooth part survives.  The two time scales involved should be distinguished: the leading envelope falls below the $O(1/N)$ correction at the crossover time $t_\times \sim \tau_\lambda N^{1/3}$, whereas the detailed spectral-correlation structure $\mathcal{F}(t/\tau_d)$ evolves on the Heisenberg scale $t \sim \tau_d$.  At such long times the Poissonian ensemble exhibits a different spectral form factor (no correlation hole) and different late-time fluctuations.  However, for all practical purposes of thermalization, the two ensembles behave identically.  This reinforces the central message of our work: the random free fermion model captures the eigenbasis-chaos features of thermalization, and its predictions for the leading relaxation are robust against modifications of the fine-grained spectral statistics at fixed one-point density and eigenvector statistics.

\section{Many-body ETH and the origin of its energy dependence}
\label{sec:conservation}
This section establishes the many-body content of the statistical decomposition of Sec.~\ref{sec:intro}: it connects the single-particle dynamics to fixed-particle-number thermodynamics, derives the energy-independent ETH structure from eigenvector-eigenvalue independence (layers (i)--(ii)), and then switches on layer (iii) through a controlled correlation deformation.

\subsection{Thermodynamics at fixed particle number}
The  analysis so far was carried out in the single‑particle sector. To connect with the  many‑body system, we must consider states with a fixed particle number $N_p$. The full many‑body Hilbert space decomposes into orthogonal sectors labeled by $N_p$, each of dimension $\binom{N}{N_p}$. Within a fixed $N_p$ sector, the Hamiltonian is still diagonal in the Slater determinant basis.

For an initial state $|\Psi_0\rangle$ with  particle number $N_p$ and a narrow energy distribution, the long‑time average of $\operatorname{Tr}(A\rho(t))$ is given by the diagonal ensemble:
\begin{equation}
\overline{\operatorname{Tr}(A\rho(t))} = \sum_{\mathcal{A}} |\langle \mathcal{A}|\Psi_0\rangle|^2 \, \langle \mathcal{A}|A|\mathcal{A}\rangle,
\end{equation}
where the sum runs over many‑body eigenstates $|\mathcal{A}\rangle$ in the $N_p$-particle sector. The ETH-type behavior of the eigenstate expectation values of one-body observables in this model \cite{Magan2016} means that each $\langle \mathcal{A}|A|\mathcal{A}\rangle$ is close to the microcanonical average at energy $E_{\mathcal{A}}$ up to fluctuations of order $N^{-1/2}$ (in the sense of typical eigenstates; the exact fluctuation scale is given by \eqref{eq:covcorr} below). Therefore,
\begin{equation}
\overline{\operatorname{Tr}(A\rho(t))} \approx \operatorname{Tr}\!\bigl( \rho_{\mathrm{eq}}(\beta,\mu) A \bigr),
\end{equation}
where $\rho_{\mathrm{eq}}(\beta,\mu)$ is the grand canonical ensemble at inverse temperature $\beta$ and chemical potential $\mu$ determined by $N_p$ and the average energy.

To make this precise, we compute the average two‑point correlation function in a typical $N_p$-particle eigenstate. From the eigenvector statistics, we have
\begin{equation}
C_{ij}^{\Psi} \equiv \langle \Psi^{N_p}| c_i^\dagger c_j |\Psi^{N_p}\rangle = \sum_{a\in\mathcal{A}} \psi_i^a \psi_j^a.
\end{equation}
Averaging over the eigenvectors gives
\begin{equation}
\langle C_{ij}^{\Psi} \rangle = \frac{N_p}{N}\,\delta_{ij} = f\delta_{ij}, \qquad f \equiv \frac{N_p}{N}.
\label{eq:avgcorr}
\end{equation}
The fluctuations follow from the exact fourth moment \eqref{eq:haar4} of the Haar measure, and gives
\begin{equation}
\langle C_{ij}^{\Psi} C_{kl}^{\Psi} \rangle - \langle C_{ij}^{\Psi}\rangle \langle C_{kl}^{\Psi}\rangle
= \frac{N_p(N-N_p)}{(N-1)N(N+2)}
\left[ \delta_{ik}\delta_{jl} + \delta_{il}\delta_{jk} - \frac{2}{N}\,\delta_{ij}\delta_{kl} \right].
\label{eq:covcorr}
\end{equation}
For the diagonal entries this yields the exact variance of the occupation number,
\begin{equation}
\operatorname{Var}\bigl( C_{ii}^{\Psi} \bigr) = \frac{2f(1-f)}{N+2},
\end{equation}
in agreement with Mag\'an's result \cite{Magan2016}. Note the factor $(1-f)$, which is absent in the independent-Gaussian (Wick) approximation of the fourth moment. As in \eqref{eq:U2psi}, the exact Haar moment is essential at this order.

Now consider the grand canonical ensemble at temperature $\beta$ and chemical potential $\mu$. Its correlation function is
\begin{equation}
C_{ij}^{\beta,\mu} = \frac{1}{Z}\mathrm{Tr}\!\bigl( e^{-\beta(H-\mu N)} c_i^\dagger c_j \bigr) = \sum_{a} \psi_i^a \psi_j^a \, n_F(E_a-\mu),
\end{equation}
with $n_F(x) = (e^{\beta x}+1)^{-1}$. Averaging over eigenvectors yields
\begin{equation}
\langle C_{ij}^{\beta,\mu} \rangle = \bigl[ n_F(E-\mu) \bigr]_{\mathrm{avg}} \, \delta_{ij},
\end{equation}
where $[\,\cdot\,]_{\mathrm{avg}}$ denotes the average over the single‑particle energy distribution $\rho(E)$. Comparing with \eqref{eq:avgcorr}, we see that the two averages coincide if
\begin{equation}
f = \bigl[ n_F(E-\mu) \bigr]_{\mathrm{avg}}.
\label{eq:match}
\end{equation}
This is a single scalar equation and therefore determines a one-parameter family of pairs $(\beta,\mu)$ consistent with a given filling fraction $f$. A second condition is supplied by the average single-particle energy density $e = E_{\mathcal{A}}/N$ of the many-body state. For the eigenstate expectation values to reproduce the grand canonical average, we must also require
\begin{equation}
e = \bigl[ E\, n_F(E-\mu) \bigr]_{\mathrm{avg}}.
\label{eq:energy_match}
\end{equation}
The two equations \eqref{eq:match} and \eqref{eq:energy_match} determine the effective temperature and chemical potential as functions of the pair $(f,e)$; within the thermodynamic domain considered here (away from the boundaries of the semicircle and the filling), these equations determine $(\beta,\mu)$ uniquely.
For instance, in the high‑temperature limit $\beta\to0$, \eqref{eq:match} gives
\begin{equation}
f \approx \frac{1}{2} - \frac{\beta}{4}(\alpha - \mu) + O(\beta^2).
\end{equation}
The remaining combination is fixed by \eqref{eq:energy_match}, which at the same order reads
\begin{equation}
e \approx \frac{\alpha}{2} - \frac{\beta}{4}\bigl(\alpha^2 + \eta^2 N - \mu\alpha\bigr) + O(\beta^2).
\end{equation}

Thus, the value approached by few-body observables in the long‑time limit is the finite‑temperature grand canonical average, not the infinite‑temperature value appearing in the ensemble average \eqref{eq:limit}. The infinite‑temperature result emerges only after additionally averaging over all energies, which washes out the energy constraint.
This reconciles the random-matrix ensemble average with the physical expectation that the local observables of a system with fixed particle number and energy are reproduced, in the thermodynamic limit, by a grand canonical ensemble at some $\beta$ and $\mu$.

\subsection{Energy-shell ETH from eigenvector--eigenvalue independence}
Recall that, in ETH, 
the eigenstate expectation values of a few-body observable, when restricted to a {\it narrow} energy shell, concentrate around a smooth microcanonical function $O_{\mathrm{th}}(E)$, with fluctuations of order $e^{-S(E)/2}$ where $S(E)$ is the thermal entropy at energy $E$. For our model both the diagonal and the off-diagonal ETH structure can be computed, leading to a precise statement about \emph{which ingredient of a Hamiltonian generates which part of ETH}. 

The central structural result of this section is the following factorization theorem, which answers the question of \emph{why} the ETH function of this model is energy independent: not because the computed average happens to be constant, but because the eigenvectors and the eigenvalues are statistically independent.
\begin{theorem}[Factorization of the eigenstate statistics]
\label{thm:factorization}
For the GOE ensemble, the eigenvector matrix $U$ and the spectrum $\{E_a\}$ are statistically independent random variables. Consequently, for every fixed occupation set $\mathcal{A}$ and every one-body observable $A$,
\begin{equation}
P\bigl(O_{\mathcal{A}},\, E_{\mathcal{A}}\bigr) = P\bigl(O_{\mathcal{A}}\bigr)\, P\bigl(E_{\mathcal{A}}\bigr).
\label{eq:factorization}
\end{equation}
exactly. Therefore the conditional average is exactly energy-independent,
\begin{equation}
\frac{\partial}{\partial E}\,\mathbb{E}\bigl[ O_{\mathcal{A}} \,\big|\, E_{\mathcal{A}} = E \bigr] = 0,
\qquad
\mathbb{E}\bigl[ O_{\mathcal{A}} \,\big|\, E_{\mathcal{A}} = E \bigr] = \frac{N_p}{N}\operatorname{Tr}A,
\end{equation}
and the conditional variance equals the unconditional one, exactly,
\begin{equation}
\operatorname{Var}\bigl( O_{\mathcal{A}} \,\big|\, E_{\mathcal{A}} = E \bigr)
= \operatorname{Var}(O_{\mathcal{A}}).
\end{equation}
\end{theorem}
For a continuous spectrum, the conditional notation is understood as the limit of conditioning on an energy shell of vanishing width (equivalently, a regular conditional probability); all statements are unaffected by this choice. Note that the statements at the level of the joint law \eqref{eq:factorization} constitute exact energy-factorization, as they are exact independence statements of the ensemble. The factorization is exact at the ensemble level and does \emph{not} imply exact independence between energy and observable values within a single realization: the statements for a \emph{fixed realization} of the Hamiltonian, with $\mathcal{A}$ ranging over the eigenstates in an energy shell, constitute \emph{energy-shell conditioning}, because they carry the finite-population covariance subtraction \eqref{eq:eth_condvar_mc} below and are \emph{not} exact independence statements. This subtraction is the bridge from the exact ensemble independence to the physical energy-shell ETH, and it is of relative order $1/N$.

\begin{proof}[Proof of Theorem \ref{thm:factorization}]
For a one-body observable $A$, denote its matrix elements in the single-particle energy basis by
\begin{equation}
A_{ab} \equiv \langle a|A|b\rangle = \sum_{ij} A_{ij}\,\psi_i^a \psi_j^b,
\label{eq:Aab}
\end{equation}
so that the many-body diagonal and off-diagonal matrix elements are respectively
\begin{equation}
O_{\mathcal{A}} \equiv \langle\mathcal{A}|A|\mathcal{A}\rangle = \sum_{a\in\mathcal{A}} A_{aa},
\qquad
\langle\mathcal{A}|A|\mathcal{B}\rangle = A_{ab}\ \ \text{iff}\ \ \mathcal{B} = (\mathcal{A}\setminus\{a\})\cup\{b\},\label{5656}
\end{equation}
with $\mathcal{B}$ differing from $\mathcal{A}$ by exactly one particle-hole excitation $(a\to b)$, and otherwise $\langle\mathcal{A}|A|\mathcal{B}\rangle = 0$. The corresponding energy difference is $\omega = E_{\mathcal{A}} - E_{\mathcal{B}} = E_a - E_b$. The off-diagonal structure is thus exactly one-to-one with the single-particle matrix elements, and we emphasize that for a one-body observable the many-body off-diagonal matrix is \emph{highly sparse}, with only $N_p(N-N_p)$ nonzero entries out of $\binom{N}{N_p}^{2}$; all statements below about off-diagonal matrix elements refer to this \emph{allowed particle-hole sector}.

\paragraph{Ensemble level.}
For the GOE the eigenvector matrix and the spectrum are statistically independent random variables. Hence, at fixed $\mathcal{A}$, the observable
$O_{\mathcal{A}} = \sum_{a\in\mathcal{A}} A_{aa}$, which is a function of the eigenvectors only, is \emph{independent} of the energy
$E_{\mathcal{A}} = \sum_{a\in\mathcal{A}} E_a$, a function of the eigenvalues only. Consequently
\begin{equation}
\mathbb{E}\bigl[ O_{\mathcal{A}} \,\big|\, E_{\mathcal{A}} = E \bigr] = \frac{N_p}{N}\,\operatorname{Tr}A,
\label{eq:eth_condmean_ens}
\end{equation}
and the conditional variance equals the unconditional one,
\begin{equation}
\operatorname{Var}\bigl( O_{\mathcal{A}} \,\big|\, E_{\mathcal{A}} = E \bigr)
= \operatorname{Var}(O_{\mathcal{A}})
= \frac{2 N_p (N-N_p)}{(N-1)N(N+2)} \left[ \operatorname{Tr}(A^2) - \frac{(\operatorname{Tr}A)^2}{N} \right],
\label{eq:eth_condvar_ens}
\end{equation}
using \eqref{eq:covcorr}. In this ensemble, therefore, energy-shell conditioning does not alter the ETH fluctuations at all.

\paragraph{Fixed realization.}
For a \emph{fixed} Hamiltonian, we consider how $O_{\mathcal{A}}$ fluctuates as $\mathcal{A}$ ranges over the $N_p$-particle eigenstates in an energy window. Let $\{(A_{aa}, E_a)\}_{a=1}^{N}$ be a finite population of exactly $N$ pairs, and $(O_{\mathcal{A}}, E_{\mathcal{A}})$ the sums of $N_p$ of them, through \eqref{5656}. Let $\operatorname{Var}_{\mathcal A}$ and $\operatorname{Cov}_{\mathcal A}$ denote the variance and covariance over the $\binom{N}{N_p}$ subsets $\mathcal{A}$ of fixed size $N_p$, at a fixed realization. 
 The exact finite-population moments follow from the elementary  identities $\sum_a A_{aa} = \operatorname{Tr}A$, $\sum_a A_{aa}\epsilon_a = \operatorname{Tr}(AV)$ and $\sum_a \epsilon_a^2 = \operatorname{Tr}(V^2)$:
\begin{align}
\operatorname{Var}_{\mathcal{A}}[O] &= \frac{N_p(N-N_p)}{N-1}\left[ \frac{1}{N}\sum_a A_{aa}^2 - \Bigl(\frac{\operatorname{Tr}A}{N}\Bigr)^2 \right],
\label{eq:eth_varA} \\
\operatorname{Var}_{\mathcal{A}}[E] &= \eta^2\,\frac{N_p(N-N_p)}{N-1}\left[ \frac{\operatorname{Tr}(V^2)}{N} - \Bigl(\frac{\operatorname{Tr}V}{N}\Bigr)^2 \right],
\label{eq:eth_varE} \\
\operatorname{Cov}_{\mathcal{A}}[O,E] &= \eta\,\frac{N_p(N-N_p)}{N-1}\left[ \frac{\operatorname{Tr}(AV)}{N} - \frac{\operatorname{Tr}A\,\operatorname{Tr}V}{N^2} \right].
\label{eq:eth_cov}
\end{align}
In the central-limit regime (large $N_p$, $N-N_p$) the joint distribution of $(O_{\mathcal{A}}, E_{\mathcal{A}})$ is bivariate Gaussian with these moments, so that
\begin{equation}
\mathbb{E}[O\,|\,E] = \frac{N_p}{N}\operatorname{Tr}A + \frac{\operatorname{Cov}_{\mathcal{A}}[O,E]}{\operatorname{Var}_{\mathcal{A}}[E]}
\left( E - N_p\alpha - \eta\,\frac{N_p}{N}\operatorname{Tr}V \right),
\label{eq:eth_condmean_mc}
\end{equation}
and
\begin{equation}
\operatorname{Var}(O\,|\,E) = \operatorname{Var}_{\mathcal{A}}[O] - \frac{\operatorname{Cov}_{\mathcal{A}}[O,E]^2}{\operatorname{Var}_{\mathcal{A}}[E]}.
\label{eq:eth_condvar_mc}
\end{equation}
The covariance subtraction in \eqref{eq:eth_condvar_mc} is the quantitative imprint of the energy constraint. Its typical magnitude is easy to estimate: using the Haar-averaged value $\sum_a A_{aa}^2 = [2\operatorname{Tr}(A^2)+(\operatorname{Tr}A)^2]/(N+2)$ (which reproduces \eqref{eq:eth_condvar_ens} at leading order), the GOE moments $\langle \operatorname{Tr}(AV)^2\rangle = 2\operatorname{Tr}(A^2)$ and $\langle \operatorname{Tr}V^2\rangle = N^2 + N$, one finds
\begin{equation}
\frac{\operatorname{Cov}_{\mathcal{A}}[O,E]^2}{\operatorname{Var}_{\mathcal{A}}[E]\,\operatorname{Var}_{\mathcal{A}}[O]}
\;\sim\; \frac{1}{N},
\label{eq:eth_sub_ratio}
\end{equation}
so that the energy-shell conditioning reduces the ETH variance only by a \emph{relative} $O(1/N)$ amount:
\begin{equation}
\operatorname{Var}(O\,|\,E) = \operatorname{Var}_{\mathcal{A}}[O]\,\bigl[\,1 + O(1/N)\,\bigr],
\qquad
\mathbb{E}[O\,|\,E] = \frac{N_p}{N}\operatorname{Tr}A + O(1/N).
\end{equation}
\end{proof}
Thus, in this model the energy-shell-conditioned variance is, at leading order in $1/N$, \emph{energy-independent and entirely eigenvector-driven}; the energy constraint acts only through the subleading covariance subtraction, and even conditioning on the shell of a specific realization modifies the fluctuations by a parametrically small amount.

Next, we consider the off-diagonal structure.
The Haar statistics of the off-diagonal single-particle matrix elements follows from the exact fourth moment \eqref{eq:haar4}: for $a\neq b$,
\begin{equation}
\bigl\langle |A_{ab}|^2 \bigr\rangle = \frac{N\,\operatorname{Tr}(A^2) - (\operatorname{Tr}A)^2}{(N-1)N(N+2)}
= \frac{\operatorname{Tr}(A^2)}{N^2}\,\bigl[\,1+O(1/N)\,\bigr].
\label{eq:eth_offdiag}
\end{equation}
 Note the exact consistency relation 
 \begin{equation}\sum_{a\neq b}\langle|A_{ab}|^2\rangle = \operatorname{Tr}(A^2) - \sum_a\langle A_{aa}^2\rangle = [N\operatorname{Tr}(A^2)-(\operatorname{Tr}A)^2]/(N+2).\end{equation} 
 To avoid confusion with the smooth $O(1)$ function $f(\bar E,\omega)$ of the standard ETH ansatz $A_{mn} = A(\bar E)\delta_{mn} + e^{-S(\bar E)/2}f(\bar E,\omega)R_{mn}$ \cite{Srednicki1994}, we define $f_A(E,\omega)$ as the \emph{energy-resolved second moment} of the off-diagonal matrix elements. Since $|A_{ab}|^2$ is a function of the eigenvectors only and the eigenvectors are independent of the spectrum, this conditional average is, for every $E$ and $\omega$,
\begin{align}
f_A(E,\omega) \;&\equiv\; \bigl\langle |A_{ab}|^2 \;\big|\; E_a = E,\; E_b = E-\omega,\; a\neq b \bigr\rangle \notag \\
&= \frac{N\,\operatorname{Tr}(A^2) - (\operatorname{Tr}A)^2}{(N-1)N(N+2)}
= \frac{\operatorname{Tr}(A^2)}{N^2}\,\bigl[\,1+O(1/N)\,\bigr],
\label{eq:eth_spectral_function}
\end{align}
with no dependence on $E$, $\omega$, or the eigenvalue statistics.
 This is the off-diagonal counterpart of Theorem~\ref{thm:factorization}: within the allowed particle-hole sector, the ETH spectral function $f_A(E,\omega)$ is a constant---an exact statement about the sparse off-diagonal sector, not about a dense many-body ETH matrix---and the \emph{nontrivial} energy-dependent structure of $f_A(E,\omega)$ arises from eigenvector-eigenvalue correlations (Sec.~\ref{sec:deformation}). We emphasize that $f_A(E,\omega)$ is a \emph{single-particle}, energy-resolved second moment: for one-body observables the many-body off-diagonal matrix elements in the allowed particle-hole sector coincide exactly with the single-particle $A_{ab}$, so $f_A$ serves as an ETH-like diagnostic of this sector, and it is not to be identified with the dense-matrix spectral function $f(\bar E,\omega)$ of a generic many-body operator.

The results so far organize all relevant quantities into four distinct layers, each tied to one statistical ingredient of the ensemble:
\begin{enumerate}
    \item \emph{Eigenvector (Haar) statistics} determine the ETH diagonal variance \eqref{eq:eth_condvar_ens} and the off-diagonal amplitude \eqref{eq:eth_spectral_function}.
    \item \emph{One-point spectral density} determines the leading thermalization profile $g(t/\tau_\lambda)$; the subleading fixed-realization energy-shell drift \eqref{eq:eth_condmean_mc} is governed by the observable-energy covariance \eqref{eq:eth_cov} of the fixed realization.
    \item \emph{Eigenvector--eigenvalue correlations} determine the genuinely energy-dependent ETH structure, generated at first order in $\varepsilon$ by the controlled deformation of Sec.~\ref{sec:deformation}.
    \item \emph{Two-level spectral correlations} determine the $O(1/N)$ two-level-correlation contribution $\mathcal{F}(t/\tau_d)$ to the single-particle spectral form factor, and the corresponding subleading late-time correction to the dynamics.
\end{enumerate}
This is the precise form, in this model, of the statement that eigenbasis chaos and spectral chaos are independent.

\begin{remark}
\label{rem:global_RI}
Theorem~\ref{thm:factorization} has a clean interpretation as a statement
about random basis transformations. The Haar-distributed eigenvector matrix
$U=(\psi_i^a)$ of the one-body GOE is precisely a \emph{global} random
rotation of the eigenbasis, because the matrix elements $A_{ab}$ of the observable in
the energy basis have the same joint statistics as those of
$U^{\dagger}AU$ for Haar-$U\in O(N)$. This is the featureless
(single-block) limit of local rotational invariance~\cite{Vallini2026}, for
which the local free cumulants reduce to the normalized free cumulants of
$A$, $K_n=\kappa_n/N^{n-1}$, and carry no energy dependence. Concretely,
\eqref{eq:eth_offdiag} is the finite-$N$ (exact Haar) refinement of the
free-cumulant identity
\begin{equation}
\bigl\langle |A_{ab}|^2\bigr\rangle=\frac{\kappa_2}{N}\bigl[1+O(1/N)\bigr],
\qquad
\kappa_2=\frac{\operatorname{Tr}(A^2)}{N}-\Bigl(\frac{\operatorname{Tr}A}{N}\Bigr)^{2},
\qquad a\neq b,
\end{equation}
cf.\ the $K_2(E_i,E_j)=\operatorname{Tr}(\hat\sigma_{E_i}A\hat\sigma_{E_j}A)
-[\operatorname{Tr}(\hat\sigma_{E_i}A)]^2\sigma(E_i-E_j)$ in \cite{Vallini2026}. The energy
independence of \eqref{eq:eth_condmean_ens} and
\eqref{eq:eth_spectral_function} is thus not a peculiarity of the
free-fermion model but the defining property of the globally
rotation-invariant limit. Conversely, in this model we can prove it,
whereas in a generic chaotic system it is an assumption to be tested.
\end{remark}

\begin{remark}
\label{rem:empirical}
A central result of Ref.~\cite{Vallini2026} is that, under (local) rotational
invariance, ensemble averages over random rotations coincide with
\emph{empirical averages} of a single deterministic operator over small
energy windows. The
leading factorization of a product with repeated indices is accompanied by a
relative $O(1/D)$ correction fixed by the off-diagonal two-point function at
zero frequency,
\begin{equation}
\bigl\langle A_{ii}^2\bigr\rangle-\bigl\langle A_{ii}\bigr\rangle^2
\simeq \frac{2}{\beta_{\mathrm{RMT}}}\,\bigl\langle |A_{ij}|^2\bigr\rangle,
\qquad
\beta_{\mathrm{RMT}}=1\ (\text{orthogonal}),\ \ 2\ (\text{unitary}),
\end{equation}
see their Eqs.~(E10)/(F1). Because our ensemble is a GOE
($\beta_{\mathrm{RMT}}=1$) with \emph{exactly} known Haar moments, this can be
verified in closed form. From \eqref{eq:haar4} one finds
\begin{equation}
\bigl\langle A_{aa}^2\bigr\rangle
=\frac{2\operatorname{Tr}(A^2)+(\operatorname{Tr}A)^2}{N(N+2)},
\qquad
\langle A_{aa}\rangle=\frac{\operatorname{Tr}A}{N},
\end{equation}
so that, using \eqref{eq:eth_offdiag},
\begin{equation}
\operatorname{Var}(A_{aa})
=\frac{2}{N(N+2)}\Bigl[\operatorname{Tr}(A^2)-\frac{(\operatorname{Tr}A)^2}{N}\Bigr]
=2\bigl\langle |A_{ab}|^2\bigr\rangle\,\frac{N-1}{N},
\qquad a\neq b .
\label{eq:class_factor}
\end{equation}
The ratio $2(N-1)/N\to 2$ is precisely the orthogonal-class factor
$2/\beta_{\mathrm{RMT}}$, obtained here with the
exact $O(N)$ Weingarten function.
\end{remark}

\subsection{Extension to $k$-body observables}
\label{sec:kbody}
The analysis above can be extended to $k$-body observables. Consider a $k$-body operator diagonal in the position basis,
\begin{equation}
O^{(k)} = \sum_{i_1<\cdots<i_k} W_{i_1\cdots i_k}\, n_{i_1}\cdots n_{i_k},
\qquad
n_i = c_i^\dagger c_i,
\end{equation}
with a symmetric weight tensor $W$, and denote by 
\begin{equation}\hat{O}_A = \sum_{i_1<\cdots<i_k} W_{i_1\cdots i_k}|\psi_{i_1}^{a_1}|^2\cdots|\psi_{i_k}^{a_k}|^2\end{equation}
the kernel of a $k$-tuple $A=\{a_1,\dots,a_k\}$ of orbitals, so that $O^{(k)}_{\mathcal{A}} = \sum_{A\subset\mathcal{A},\,|A|=k}\hat{O}_A$.

 Two structural results hold for every fixed $k$:
\begin{enumerate}
\item \emph{Exact variance.} With $T^{(k)} = \sum_{|A|=k}\hat{O}_A\,I_A$ and $p_k = \binom{N_p}{k}/\binom{N}{k}$, the variance over subsets is the finite-population sum
\begin{equation}
\operatorname{Var}_{\mathcal{A}}\bigl[T^{(k)}\bigr]
= \sum_{j=0}^{k} \bigl( p_{2k-j} - p_k^{2} \bigr)\, S^{(j)},
\qquad
S^{(j)} = \sum_{|A\cap B|=j} \hat{O}_A \hat{O}_B,
\label{eq:k_gen}
\end{equation}
which follows from the elementary covariances $\operatorname{Cov}_{\mathcal{A}}(I_A,I_B) = p_{|A\cup B|}-p_{|A|}\,p_{|B|}$ of the subset indicators.

\item \emph{Energy covariance.} The covariance with the many-body energy carries the coefficient
\begin{equation}
p_{k+1} - f\,p_k \;=\; -\frac{k f^{k}(1-f)}{N}\,\bigl[\,1+O(1/N)\,\bigr],
\qquad \text{$k$ and $f$ fixed as $N\to\infty$},
\label{eq:k_pk}
\end{equation}
so that, since $\operatorname{Var}_{\mathcal{A}}[E] \sim \eta^{2} f(1-f) N^{2}$,
\begin{equation}
\frac{\operatorname{Cov}_{\mathcal{A}}\bigl[T^{(k)},E\bigr]^{2}}
{\operatorname{Var}_{\mathcal{A}}[E]\,\operatorname{Var}_{\mathcal{A}}\bigl[T^{(k)}\bigr]}
= O(1/N)
\qquad\text{for every fixed } k.
\label{eq:k_gen_sub}
\end{equation}
The energy-shell conditioning of the $k$-body ETH fluctuations is therefore always a subleading, $1/N$-suppressed effect.
\end{enumerate}

The two-body case provides the coefficient-level version of \eqref{eq:k_gen}. Consider $O^{(2)} = \sum_{i<j} W_{ij}\, n_i n_j$ with a symmetric weight extended by $W_{ji}=W_{ij}$ and $W_{ii}=0$, and define the symmetrized two-orbital kernel
\begin{equation}
\hat{O}_{ab} = \sum_{i,j} W_{ij}\,|\psi_i^a|^2|\psi_j^b|^2,
\qquad
T \;=\; \sum_{a,b\in\mathcal{A}} \hat{O}_{ab} = 2\,O^{(2)}_{\mathcal{A}},
\label{eq:k2_def}
\end{equation}
so that $\hat{O}_{ab} = \hat{O}_{ba}$ and $T$ equals twice the diagonal element of the observable. With the inclusion probabilities
\begin{equation}
p_k = \frac{\binom{N_p}{k}}{\binom{N}{k}},
\label{eq:pk}
\end{equation}
the exact one-realization sums
\begin{equation}
d_a = \hat{O}_{aa},\quad D = \sum_a d_a,\quad m_a = \sum_{b\neq a}\hat{O}_{ab},\quad M = \sum_{a,b}\hat{O}_{ab},
\label{eq:k2_dm}
\end{equation}
\begin{equation}
S_{\mathrm{dd}} = \sum_a d_a^{2},\quad S_{2} = \sum_{a<b}\hat{O}_{ab}^{2},\quad S_{\mathrm{dp}} = \sum_{a\neq b} d_a\hat{O}_{ab},
\label{eq:k2_S}
\end{equation}
\begin{equation}
S_{1} = \sum_{a<b<c} \bigl( \hat{O}_{ab}\hat{O}_{ac} + \hat{O}_{ab}\hat{O}_{bc} + \hat{O}_{ac}\hat{O}_{bc} \bigr),
\qquad
S_{\mathrm{dg}} = \sum_{a<b<c} \bigl( d_a\hat{O}_{bc} + d_b\hat{O}_{ac} + d_c\hat{O}_{ab} \bigr),
\label{eq:k2_S1}
\end{equation}
and $S_0$ fixed by the identity
\begin{equation}
M^{2} = S_{\mathrm{dd}} + 4S_{2} + \bigl(D^{2}-S_{\mathrm{dd}}\bigr) + 4S_{\mathrm{dp}} + 4S_{\mathrm{dg}} + 8S_{1} + 8S_{0},
\label{eq:k2_M2}
\end{equation}
the variance over subsets is, 
\begin{align}
\operatorname{Var}_{\mathcal{A}}\bigl[T\bigr]
&= (p_1-p_1^{2})S_{\mathrm{dd}} + 4(p_2-p_2^{2})S_{2} + (p_2-p_1^{2})\bigl(D^{2}-S_{\mathrm{dd}}\bigr) \notag \\
&\quad + 4(p_2-p_1p_2)S_{\mathrm{dp}} + 4(p_3-p_1p_2)S_{\mathrm{dg}} + 8(p_3-p_2^{2})S_{1} + 8(p_4-p_2^{2})S_{0},
\label{eq:k2_var_exact}
\end{align}
and the covariance with the many-body energy is, 
\begin{align}
\operatorname{Cov}_{\mathcal{A}}\bigl[T,\, E\bigr]
&= (p_1-p_1^{2})\sum_a \varepsilon_a d_a + (p_2-p_1^{2})\sum_a \varepsilon_a\bigl(D-d_a\bigr) \notag \\
&\quad + 2(p_2-p_1p_2)\sum_a \varepsilon_a m_a + (p_3-p_2p_1)\sum_a \varepsilon_a\bigl(M-D-2m_a\bigr),
\label{eq:k2_cov}
\end{align}
while $\operatorname{Var}_{\mathcal{A}}[E] = \eta^{2}\bigl[(p_1-p_2)\sum_a\varepsilon_a^{2} + (p_2-p_1^{2})\bigl(\sum_a\varepsilon_a\bigr)^{2}\bigr]$.
In the bivariate-Gaussian regime (valid as $N_p, N-N_p\to\infty$ at fixed $k$ for regular weights) the conditional variance is therefore
\begin{equation}
\operatorname{Var}\bigl(T\,\big|\,E\bigr)
= \operatorname{Var}_{\mathcal{A}}\bigl[T\bigr] - \frac{\operatorname{Cov}_{\mathcal{A}}\bigl[T,E\bigr]^{2}}{\operatorname{Var}_{\mathcal{A}}\bigl[E\bigr]},
\label{eq:k2_condvar_exact}
\end{equation}
the two-body analogue of \eqref{eq:eth_condvar_mc} at full coefficient level. 

The Haar content of the coefficient sums follows from the exact same-column moments $\langle|\psi_i^a|^2|\psi_k^a|^2\rangle = (1+2\delta_{ik})/[N(N+2)]$ and, to leading order in $1/N$, the cross-column moment $\langle|\psi_i^a|^2|\psi_j^b|^2\rangle = 1/[N(N+2)]$ for $a\neq b$. Defining $\overline{W} = \sum_{i,j}W_{ij}$, $R_1 = \sum_{i<j<l}(W_{ij}W_{il}+W_{ij}W_{jl}+W_{il}W_{jl})$ and $R_2 = \sum_{i<j}W_{ij}^{2}$ over the symmetrized weight, one finds, to leading order in $1/N$,
\begin{equation}
\bigl\langle \hat{O}_{ab}\bigr\rangle = \frac{\overline{W}}{N(N+2)},
\qquad
\bigl\langle \hat{O}_{ab}^{2}\bigr\rangle = \frac{\overline{W}^{2} + 2R_1 + 4R_2}{N^{2}(N+2)^{2}},
\qquad
\bigl\langle \hat{O}_{ab}\hat{O}_{ac}\bigr\rangle = \frac{\overline{W}^{2} + R_1}{N^{2}(N+2)^{2}},
\label{eq:k2_haar}
\end{equation}
together with 
\begin{align}
\langle\hat{O}_{aa}^{2}\rangle &= \frac{\overline{W}^{2}}{N(N+2)(N+4)(N+6)}\bigl[1+O(1/N)\bigr],\nonumber\\
\langle\hat{O}_{aa}\hat{O}_{ab}\rangle &= \frac{\overline{W}^{2}+2R_1}{N^{2}(N+2)(N+4)}\bigl[1+O(1/N)\bigr]
\end{align}
 for the diagonal pieces.

The coefficient-level energy imprint is read off from the  asymptotics of the $p$-combinations: $p_2-p_1^{2} = -f(1-f)/(N-1)$ and $p_3-fp_2 = -2f^{2}(1-f)/N\,\bigl[1+O(1/N)\bigr]$. Consequently
\begin{equation}
\operatorname{Cov}_{\mathcal{A}}\bigl[T,E\bigr]
= f(1-f)\sum_a\varepsilon_a d_a + 2f^{2}(1-f)\sum_a\varepsilon_a m_a + O(1/N),
\label{eq:k2_cov_lead}
\end{equation}
and, since $\operatorname{Var}_{\mathcal{A}}[E]\sim \eta^{2}f(1-f)N^{2}$, the subtraction term in \eqref{eq:k2_condvar_exact} is of relative order $1/N$ with respect to $\operatorname{Var}_{\mathcal{A}}[T]$. 

Figure~\ref{fig:kbody} (left panel) shows the measured fluctuations of the one-body two-site density $n_1+n_2$, the two-body observable $\sum_i n_i n_{i+1}$ and the three-body observable $\sum_i n_i n_{i+1}n_{i+2}$. For this specific class of translation-invariant observables with the normalization used here, the squared coefficient of variation exhibits the scalings $N^{-1}$ ($k=1$) and $N^{-3}$ ($k=2,3$) in the studied range $32\le N\le 256$. We do not claim these exponents as a universal law. The relative variance of a $k$-body observable depends on the number of terms in $O^{(k)}$, the scaling of the individual kernels, and the correlations between overlapping terms, and the exact value in each case is given by the finite-population formula \eqref{eq:k_gen} together with the relevant Haar moments. What is model-independent within this framework, for every fixed $k$ in the central-limit regime, is the $1/N$ suppression of the energy-shell subtraction \eqref{eq:k_gen_sub}. The figure thus serves as a numerical verification of the general finite-population framework, rather than as the proposal of new universal scaling exponents.
\begin{figure}[t]
\centering
\includegraphics[width=0.92\textwidth]{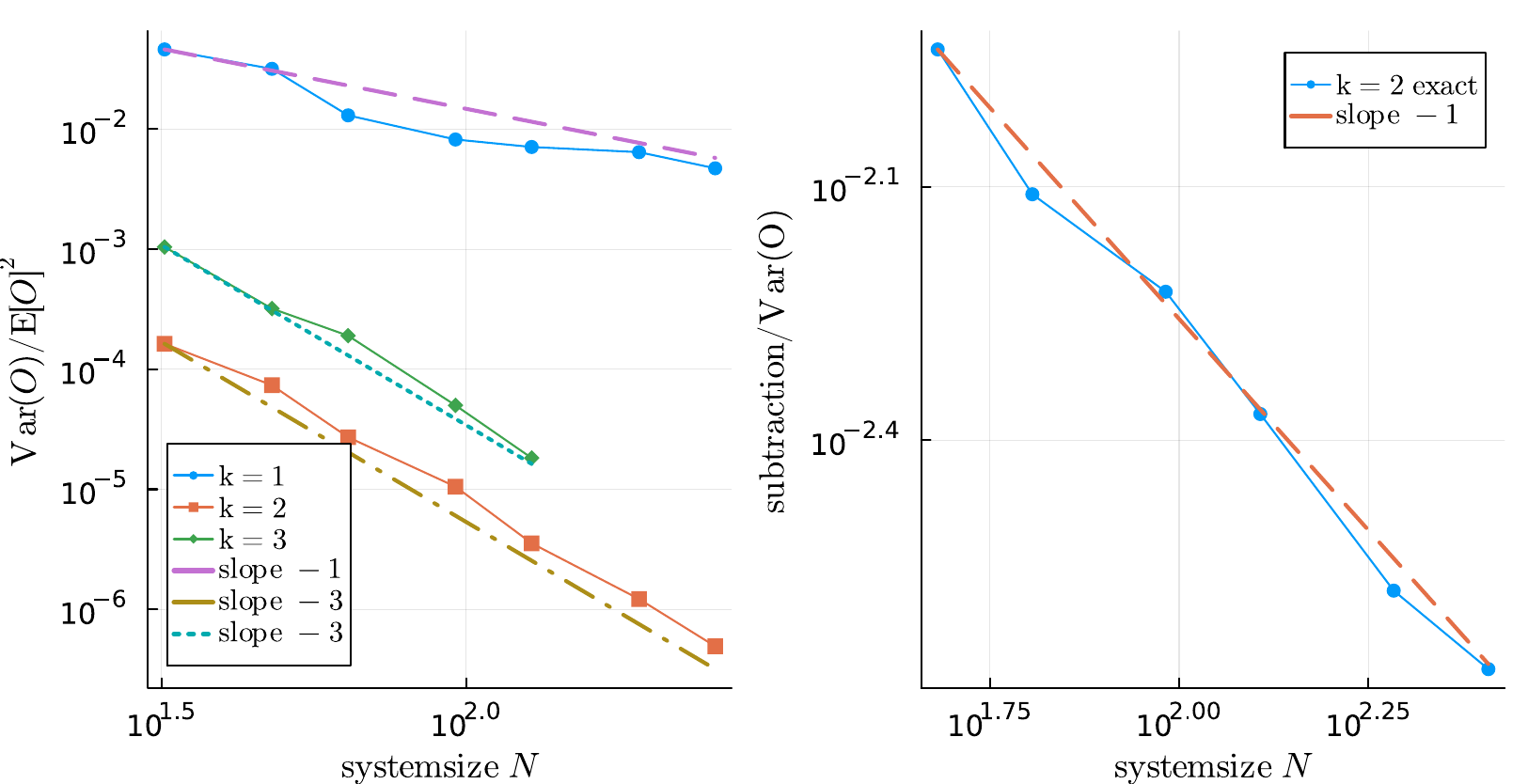}
\caption{Left: squared coefficient of variation $\operatorname{Var}(O^{(k)})/\mathbb{E}[O^{(k)}]^{2}$ versus $N$ (log-log, $N_p=N/2$) for the one-body two-site density (circles), the two-body nearest-neighbour observable $\sum_i n_i n_{i+1}$ (squares) and the three-body observable $\sum_i n_i n_{i+1}n_{i+2}$ (diamonds); the dashed and dotted lines have slopes $-1$ and $-3$. The measured scalings ($N^{-1}$ for $k=1$ and $N^{-3}$ for $k=2,3$) are specific to this class of normalized translation-invariant observables; the model-independent parametric statement is the exact finite-population formula \eqref{eq:k_gen} and the $1/N$ suppression of the energy-shell subtraction. Right: the average covariance subtraction $\langle\operatorname{Cov}_{\mathcal{A}}[T,E]^{2}/\operatorname{Var}_{\mathcal{A}}[E]\rangle/\langle\operatorname{Var}_{\mathcal{A}}[T]\rangle$ of the two-body observable, computed from the exact formulas \eqref{eq:k2_var_exact}--\eqref{eq:k2_cov} and averaged over GOE realizations, together with the slope $-1$ reference: the subtraction is $\approx 0.56/N$, the same $1/N$ suppression as for one-body observables.}
\label{fig:kbody}
\end{figure}

\subsection{Generating energy-dependent ETH by controlled correlations}
\label{sec:deformation}
Theorem~\ref{thm:factorization} and \eqref{eq:eth_spectral_function} establish that eigenvector isotropy alone generates no energy dependence of the ETH statistics.  The natural converse question---where the nontrivial energy dependence of $O_{\mathrm{th}}(E)$ and of $f_A(E,\omega)$ actually comes from---can be answered constructively within the model by switching on a controlled \emph{eigenvector-eigenvalue correlation}. 

Let us consider the deformed one-body Hamiltonian
\begin{equation}
H(\varepsilon) = H_0 + \varepsilon H_1, \qquad H_0 = \eta V \in \mathrm{GOE}, \qquad H_1 = H_1^{\mathsf{T}}\ \text{fixed},
\label{eq:deform_H}
\end{equation}
with $A$ and $H_1$ real symmetric.  At $\varepsilon=0$ the model factorizes (Theorem~\ref{thm:factorization}); for $\varepsilon>0$ the eigenvectors acquire an explicit dependence on the spectrum.  Throughout this subsection, perturbation theory is used only as an ensemble-level asymptotic expansion in $\varepsilon$: all analytical statements are first order in $\varepsilon$, while all finite-$\varepsilon$ numerical data are obtained from exact diagonalization. The perturbative statements concern the ensemble-averaged observables evaluated below; they do not assert uniform convergence of the eigenvector expansion over all realizations, since the GOE minimum level spacing is much smaller than the typical spacing. For sufficiently small $\varepsilon$, the first-order (Rayleigh-Schr\"odinger) perturbation theory gives
\begin{equation}
E_a(\varepsilon) = E_a + \varepsilon\,(H_1)_{aa} + O(\varepsilon^2), \qquad
|a(\varepsilon)\rangle = |a\rangle + \varepsilon \sum_{c\neq a} \frac{(H_1)_{ca}}{E_a - E_c}\,|c\rangle + O(\varepsilon^2),
\label{eq:deform_pt}
\end{equation}
where $(H_1)_{ca} = \langle c|H_1|a\rangle$ in the unperturbed basis, and hence
\begin{equation}
A_{aa}(\varepsilon) = A_{aa} + 2\varepsilon \sum_{c\neq a} \frac{A_{ac}\,(H_1)_{ac}}{E_a - E_c} + O(\varepsilon^2),
\label{eq:deform_dA_diag}
\end{equation}
\begin{equation}
A_{ab}(\varepsilon) = A_{ab} + \varepsilon\Bigl[ \sum_{c\neq a} \frac{A_{cb}\,(H_1)_{ca}}{E_a - E_c} + \sum_{c\neq b} \frac{A_{ac}\,(H_1)_{bc}}{E_b - E_c} \Bigr] + O(\varepsilon^2).
\label{eq:deform_dA_off}
\end{equation}
The relevant Haar ingredients follow from the exact fourth moment \eqref{eq:haar4} [for the first two of \eqref{eq:deform_haar2}] and from the leading Wick contraction of three distinct columns 
\begin{align}
\langle A_{ac}\,(H_1)_{ac}\rangle &= \frac{N\operatorname{Tr}(AH_1)-\operatorname{Tr}A\,\operatorname{Tr}H_1}{(N-1)N(N+2)} \qquad (a\neq c), \notag \\
\langle A_{aa}\,(H_1)_{aa}\rangle &= \frac{2\operatorname{Tr}(AH_1)+\operatorname{Tr}A\,\operatorname{Tr}H_1}{N(N+2)},
\label{eq:deform_haar2}
\end{align}
\begin{equation}
\langle A_{ab}\,A_{cb}\,(H_1)_{ca}\rangle = \frac{\operatorname{Tr}(A^2 H_1)}{N^3}\,\bigl[1+O(1/N)\bigr] \qquad (a,b,c\ \text{distinct}).
\label{eq:deform_haar3}
\end{equation}

Conditioning \eqref{eq:deform_dA_diag} on $E_a(\varepsilon)=E$ and using \eqref{eq:deform_haar2} gives
\begin{equation}
\mathbb{E}\bigl[ A_{aa}(\varepsilon) \,\big|\, E_a(\varepsilon) = E \bigr]
= \frac{\operatorname{Tr}A}{N} + 2\varepsilon\,\langle A_{ac}(H_1)_{ac}\rangle \sum_{c\neq a} \frac{1}{E-E_c} + O(\varepsilon^2).
\label{eq:deform_condmean}
\end{equation}
For $E$ in the spectral bulk (away from the semicircle edges) the principal-value sum is evaluated at the ensemble level. Because the eigenvectors are statistically independent of the spectrum, the spectral average factorizes from the eigenvector average, and the average of the sum over $c\neq a$ with $E_a$ held at $E$ is
 \begin{equation}\mathbb{E}\Bigl[ \sum_{c\neq a}\frac{1}{E-E_c}\;\Big|\; E_a=E \Bigr] = \mathrm{PV}\!\int dE'\frac{\rho(E')}{E-E'}\,\bigl[1+O(1/N)\bigr] = \frac{E}{2\eta^2}\,\bigl[1+O(1/N)\bigr],\end{equation} 
which is the Hilbert transform of the semicircle \eqref{eq:rhoE}. This replacement is an ensemble-level asymptotic relation in the spectral bulk; it makes no statement about the realization-to-realization fluctuations of the denominator sum, which are irrelevant for the ensemble averages computed here, and the numerical verification below covers the full band including the edge regions.  The energy-dependent \emph{single-particle energy-resolved} slope therefore emerges linearly in $\varepsilon$, at leading order in $1/N$:
\begin{equation}
\frac{\partial}{\partial E}\,\mathbb{E}\bigl[ A_{aa}(\varepsilon) \,\big|\, E_a(\varepsilon)=E \bigr]
= \frac{\varepsilon\,\operatorname{Tr}(AH_1)}{\eta^2 N^2}\,\bigl[1+O(1/N)\bigr] + O(\varepsilon^2).
\label{eq:deform_slope}
\end{equation}

\paragraph{Eigenvector--eigenvalue covariance.}
To first order in $\varepsilon$ the correlation generated by the deformation has the exact coefficient
\begin{equation}
\operatorname{Cov}\bigl( A_{aa}(\varepsilon),\, E_a(\varepsilon) \bigr)
= \varepsilon\Bigl[ \frac{\operatorname{Tr}(AH_1)}{N} - \frac{\operatorname{Tr}A\,\operatorname{Tr}H_1}{N^2} \Bigr] + O(\varepsilon^2).
\label{eq:deform_cov}
\end{equation}
Indeed $\operatorname{Cov}(A_{aa},\delta E_a) = \langle A_{aa}(H_1)_{aa}\rangle - \operatorname{Tr}A\,\operatorname{Tr}H_1/N^2$, and
\begin{equation}
\operatorname{Cov}(\delta A_{aa},E_a) = 2\sum_{c\neq a}\langle A_{ac}(H_1)_{ac}\rangle\,\Bigl\langle \frac{E_a}{E_a-E_c}\Bigr\rangle = (N-1)\,\langle A_{ac}(H_1)_{ac}\rangle,
\end{equation}
where $\langle E_a/(E_a-E_c)\rangle = 1/2$  by level-exchange symmetry.  The two pieces combine into \eqref{eq:deform_cov}.  
The deformed eigenvector matrix $U(\varepsilon)$ and the deformed spectrum $\{E_a(\varepsilon)\}$ are therefore no longer statistically independent. 

 At the many-body level the cross-covariances $\operatorname{Cov}(A_{aa}(\varepsilon), E_b(\varepsilon)) = -\varepsilon\langle A_{ac}(H_1)_{ac}\rangle + O(\varepsilon^2)$ ($b\neq a$) give
\begin{equation}
\operatorname{Cov}\bigl( O_{\mathcal A}(\varepsilon),\, E_{\mathcal A}(\varepsilon) \bigr)
= \varepsilon\, f(1-f)\,\operatorname{Tr}(AH_1)\,\bigl[1+O(1/N)\bigr] + O(\varepsilon^2),
\end{equation}
and, with $\langle\operatorname{Var}_{\mathcal A}[E]\rangle = \eta^2 N_p(N-N_p+2)\,\bigl[1+O(1/N)\bigr]$ on average over $V$, the many-body microcanonical slope
\begin{equation}
\frac{\partial}{\partial E_{\mathcal A}}\,\mathbb{E}\bigl[ O_{\mathcal A}(\varepsilon) \,\big|\, E_{\mathcal A}=E \bigr]
= \frac{\varepsilon\,\operatorname{Tr}(AH_1)}{\eta^2 N^2}\,\bigl[1+O(1/N)\bigr] + O(\varepsilon^2).
\end{equation}
We see that the negative cross-covariances cancel the $N_p$-growth of the diagonal ones, so the many-body microcanonical slope has the same coefficient as the single-orbital slope \eqref{eq:deform_slope}.

\paragraph{Off-diagonal ETH spectral function.}
From \eqref{eq:deform_dA_off} and \eqref{eq:deform_haar3}, we have
\begin{align}
\bigl\langle |A_{ab}(\varepsilon)|^2 - |A_{ab}|^2 \bigr\rangle
&= 2\varepsilon\,\langle A_{ab}\,A_{cb}\,(H_1)_{ca}\rangle \Bigl[ \sum_{c\neq a}\frac{1}{E_a-E_c} + \sum_{c\neq b}\frac{1}{E_b-E_c} \Bigr] + O(\varepsilon^2) \notag \\
&= \frac{\varepsilon\,\operatorname{Tr}(A^2H_1)\,(E_a+E_b)}{\eta^2 N^3}\,\bigl[1+O(1/N)\bigr] + O(\varepsilon^2),
\end{align}
so that, with $E_b = E-\omega$, the ETH spectral function becomes
\begin{equation}
f_A(E,\omega) = \frac{N\operatorname{Tr}(A^2)-(\operatorname{Tr}A)^2}{(N-1)N(N+2)}
+ \frac{\varepsilon\,\operatorname{Tr}(A^2H_1)\,(2E-\omega)}{\eta^2 N^3}\,\bigl[1+O(1/N)\bigr] + O(\varepsilon^2).
\label{eq:deform_fA}
\end{equation}
So, the constant of \eqref{eq:eth_spectral_function} acquires the simplest nontrivial energy dependence, linear in $E$ and $\omega$ and proportional to $\varepsilon\operatorname{Tr}(A^2H_1)$.

This provides a \emph{controlled converse mechanism} to Theorem~\ref{thm:factorization}, as we do not claim that eigenvector-eigenvalue correlations are the only possible source of energy dependence.  
The random free fermion model is the $\varepsilon=0$ baseline; the deformation \eqref{eq:deform_H} generates, through the energy denominators $(E_a-E_c)^{-1}$,  the eigenvector-eigenvalue correlation \eqref{eq:deform_cov}, and simultaneously the diagonal energy dependence \eqref{eq:deform_slope} and the off-diagonal structure \eqref{eq:deform_fA}.  
The three quantities carry a clean hierarchy of precisions: the covariance \eqref{eq:deform_cov} is exact to first order in $\varepsilon$ with an exact coefficient; the slope \eqref{eq:deform_slope} and the spectral function \eqref{eq:deform_fA} are first order in $\varepsilon$ at leading order in $1/N$, their $[1+O(1/N)]$ factors originating respectively from the ensemble-level replacement of the energy-denominator sum by its Hilbert transform and from the leading Wick contraction \eqref{eq:deform_haar3} of three distinct columns.  

\begin{remark}
\label{rem:deform_Uloc}
The perturbation \eqref{eq:deform_H} provides a perturbative microscopic
mechanism that is qualitatively consistent with the physical picture
behind local rotational invariance~\cite{Vallini2026}. The first-order
perturbation theory \eqref{eq:deform_pt} admixes $|a\rangle$ with $|c\rangle$
through the amplitude $\varepsilon(H_1)_{ac}/(E_a-E_c)$: eigenstates close in
energy, $|E_a-E_c|\lesssim\varepsilon|(H_1)_{ac}|$, are reshuffled most
strongly, so the deformation acts with an effective energy window set by
$\varepsilon$, while distant states are only weakly affected at this order.
The energy denominators $(E_a-E_c)^{-1}$ are thus the
microscopic origin of the smooth energy dependence of the ETH functions.
Concretely, comparing \eqref{eq:deform_condmean} with
$K_1(E)=\operatorname{Tr}(\hat\sigma_E A)$ and \eqref{eq:deform_fA} with
$K_2(E_i,E_j)=\operatorname{Tr}(\hat\sigma_{E_i}A\hat\sigma_{E_j}A)-
[\operatorname{Tr}(\hat\sigma_{E_i}A)]^2\sigma(E_i-E_j)$, the $\varepsilon=0$
model is the single-block (global) limit with constant $K_1,K_2$, whereas for
$\varepsilon>0$ both acquire an explicit energy dependence:
\begin{equation}
\delta K_1(E)=\frac{\varepsilon\operatorname{Tr}(AH_1)}{\eta^2N^2}\,E\bigl[1+O(1/N)\bigr],
\end{equation}
\begin{equation}
\delta K_2(E,\omega)=\frac{\varepsilon\operatorname{Tr}(A^2H_1)(2E-\omega)}{\eta^2N^3}\bigl[1+O(1/N)\bigr],
\end{equation}
at first order in $\varepsilon$ and to leading order in $1/N$. The deformation
therefore realizes, in a solvable model, a perturbative mechanism for the
transition from global to local rotational invariance; it does not construct a
strict local-Haar ensemble with a sharp energy window.
\end{remark}

We verified \eqref{eq:deform_cov}--\eqref{eq:deform_fA} numerically at $N=100$, $\eta=1$, over $10^{4}$ GOE realizations, with $A$ the two-site density $A=\frac12(|1\rangle\langle1|+|2\rangle\langle2|)$ and $H_1 = 2\,(|1\rangle\langle1|+|2\rangle\langle2|)$, for which $\operatorname{Tr}(AH_1)=2$ and $\operatorname{Tr}(A^2H_1)=1$.  For this rank-two $H_1$ the typical eigenbasis matrix elements of $H_1$ are of order $1/N$, so the typical eigenvector-mixing parameter is $\varepsilon|(H_1)_{ac}|/|E_a-E_c| = O(\varepsilon/\sqrt N)\approx 3\times10^{-2}$ at $N=100$, $\varepsilon=0.3$; accordingly we use \eqref{eq:deform_pt}--\eqref{eq:deform_dA_off} as an ensemble-level asymptotic expansion in $\varepsilon$, while all finite-$\varepsilon$ data reported below are obtained by exact diagonalization.  
Fig.~\ref{fig:deformation}, left, shows the corresponding flat-versus-linear binned conditional means: the binned $\mathbb{E}[A_{aa}(\varepsilon)|E]-\operatorname{Tr}A/N$ is flat at $\varepsilon=0$ and linear at $\varepsilon=0.3$, with a fitted slope $0.99$ times the first-order prediction $\varepsilon\operatorname{Tr}(AH_1)E/(\eta^2N^2)$; the observed ratio is consistent with the exact counting factor $(N-1)/N$ of the Hilbert-transform replacement (the sum in \eqref{eq:deform_slope} excludes exactly one level), whose remaining finite-$N$ corrections are $O(1/N)$ and below the statistical resolution at $N=100$. The binned off-diagonal change $\bigl[|A_{ab}(\varepsilon)|^2-|A_{ab}|^2\bigr]/\varepsilon$ follows the linear prediction $\operatorname{Tr}(A^2H_1)(E_a+E_b)/(\eta^2N^3)$ with a fitted slope $0.90$ times the predicted one, see Fig.~\ref{fig:deformation}, middle; an $N$-scan at $\varepsilon=0.3$ ($N=50,100,200$) gives the ratios $0.80,0.90,0.94$, i.e.\ a deficit decreasing roughly as $10/N$, consistent with an $O(1/N)$ finite-size correction to the leading three-column Haar contraction \eqref{eq:deform_haar3}. To test the first-order parameter dependence, Fig.~\ref{fig:deformation}, right, shows the ratio of the fitted slope to the first-order prediction for $\varepsilon\in\{0.1,0.2,0.3\}$ in both channels: the ratios are constant in $\varepsilon$ within the statistical errors (diagonal $0.99$, off-diagonal $0.90$), as required by a first-order-in-$\varepsilon$ structure with $\varepsilon$-independent $1/N$ corrections. The $O(\varepsilon^2)$ corrections appear instead in the curvature of the binned data: the center-bin residuals of the exact data against the first-order theory scale approximately as $1:4:9$ across $\varepsilon=0.1,0.2,0.3$, consistent with a quadratic finite-$\varepsilon$ correction. Each panel averages $10^{4}$ realizations at $N=100$ (about $4\times10^{4}$ diagonal and $2\times10^{6}$ off-diagonal samples per bin); the error bars in the right panel are the standard errors of the weighted slope fits, far below the displayed features.

\begin{figure}[t]
\centering
\includegraphics[width=0.98\textwidth]{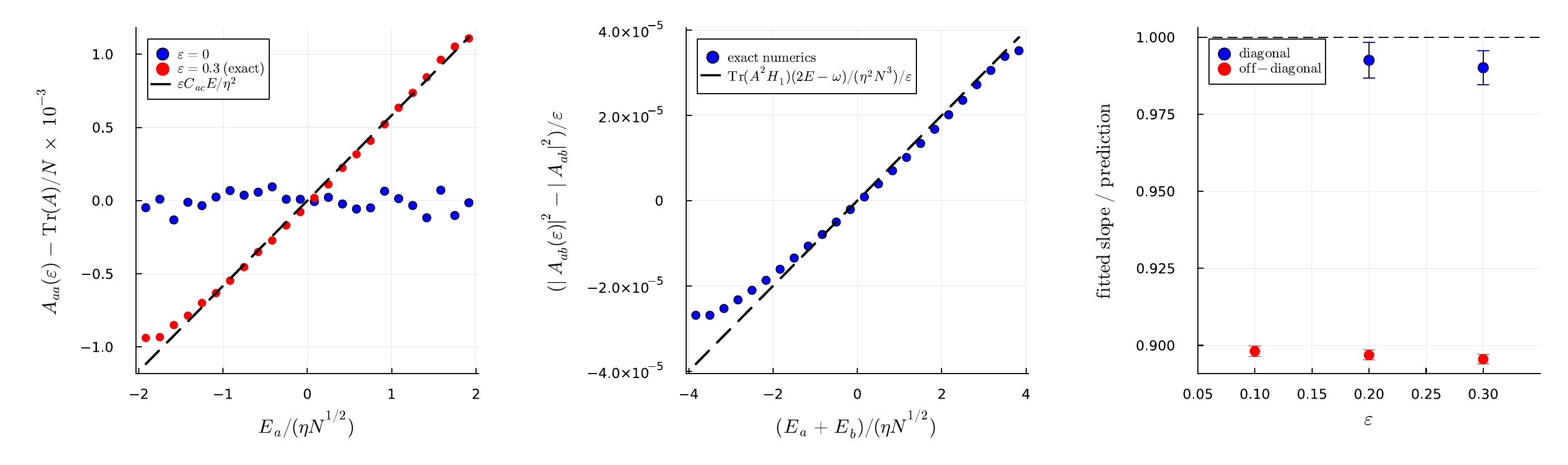}
\caption{Emergence of energy-dependent ETH structure under the eigenvector--eigenvalue correlation deformation $H(\varepsilon)=H_0+\varepsilon H_1$ ($N=100$, $\eta=1$, $10^{4}$ GOE realizations, $A$ a two-site density, $H_1=2(|1\rangle\langle1|+|2\rangle\langle2|)$). Left: binned $\mathbb{E}[A_{aa}(\varepsilon)|E]-\operatorname{Tr}A/N$ versus $E/(\eta\sqrt{N})$ for $\varepsilon=0$ (circles, flat) and $\varepsilon=0.3$ (squares, linear), together with the first-order prediction $\varepsilon\operatorname{Tr}(AH_1)E/(\eta^2N^2)$ (dashed). Middle: binned $\bigl[|A_{ab}(\varepsilon)|^2-|A_{ab}|^2\bigr]/\varepsilon$ versus $(E_a+E_b)/(\eta\sqrt{N})$ at $\varepsilon=0.3$, together with the first-order prediction $\operatorname{Tr}(A^2H_1)(E_a+E_b)/(\eta^2N^3)$ (dashed). Right: ratio of the fitted slope to the first-order prediction versus $\varepsilon$ for the diagonal channel (circles) and the off-diagonal channel (squares), with standard errors of the weighted fits; the horizontal line is the first-order prediction. The ratios are $\varepsilon$-independent (diagonal $\approx0.99$, consistent with the counting factor $(N-1)/N$ of the Hilbert-transform replacement; off-diagonal $\approx0.90$, consistent with an $O(1/N)$ correction to the three-column Haar contraction), confirming the linear-in-$\varepsilon$ generation of energy dependence; the $\varepsilon^2$ corrections appear in the curvature, with center-bin residuals scaling $1:4:9$ across $\varepsilon=0.1,0.2,0.3$.}
\label{fig:deformation}
\end{figure}

\section{Dynamics of correlation functions and spectral form factors}
\label{sec:dynamics}
This section provides a dynamical probe of the statistical decomposition established above: it re-projects the statistical mechanisms of Secs.~\ref{sec:thermalization} and \ref{sec:conservation} onto time-dependent observables. We first analyze the fluctuations of two-point correlation functions, which combine the Haar eigenvector statistics (layer (i)) with the relaxation envelope (layer (ii)), and then the single-particle spectral form factor, which isolates the two-level spectral correlations (layer (iv)). The entanglement entropy of random free fermions has been studied in the literature \cite{Magan2016,Burmistrova2017,PPGM23,Pastur2024}, so we do not repeat those computations here.

\subsection{Fluctuations of two-point functions}
\label{sec:twopoint_fluct}

For a single‑particle observable $A$, the expectation value in a time‑evolved state $|\Psi(t)\rangle$ is
\begin{equation}
\langle A(t) \rangle = \sum_{i,j} A_{ij} \, C_{ij}(t), \qquad C_{ij}(t) = \langle \Psi(t)| c_i^\dagger c_j |\Psi(t)\rangle.
\end{equation}
The correlation matrix $C(t)$ evolves as
\begin{equation}
C(t) = e^{-iht} \, C(0) \, e^{iht},
\label{eq:Cevol}
\end{equation}
where $h$ is the single‑particle Hamiltonian with matrix elements $h_{ij} = \alpha\delta_{ij} + \eta V_{ij}$. The initial state $|\Psi(0)\rangle$ is chosen to have a fixed particle number $N_p$ and a narrow energy distribution. For concreteness, we take $C(0)$ to be a rank-$N_p$ projection onto $N_p$ occupied orbitals drawn \emph{independently of} $V$ (equivalently, a Haar-random projection of rank $N_p$): the ensemble average is then taken over both the Hamiltonian and the independent choice of the occupied orbitals, so that $\langle C_{kl}(0)\rangle = f\,\delta_{kl}$ with $f = N_p/N$, and $C(0)$ is statistically independent of the time evolution.

To compute the fluctuations of $\langle A(t) \rangle$, we need the covariance of $C_{ij}(t)$. Using \eqref{eq:Cevol}, we have
\begin{equation}
C_{ij}(t) = \sum_{k,l} \bigl( e^{-iht} \bigr)_{ik} C_{kl}(0) \bigl( e^{iht} \bigr)_{lj}.
\end{equation}
Averaging over the random matrix $h$ (i.e., over $V$) involves both the eigenvector and eigenvalue statistics. Since $e^{-iht}$ is just $U(t)$ in the single‑particle basis, its moments have been computed in Sec.~\ref{3.1}. Using the second moment \eqref{eq:U2psi} together with the unitarity relation $\sum_k U_{ik}(t) U_{kj}^*(t) = \delta_{ij}$, one finds
\begin{equation}
\langle C_{ij}(t) \rangle = \sum_{k,l} \langle U_{ik}(t) U_{lj}^*(t) \rangle\, f\,\delta_{kl} = f\, \delta_{ij},
\label{eq:Cavg_exact}
\end{equation}
independently of $t$: the time dependence cancels \emph{identically} between the disconnected part of the second moment and its connected part, since the coefficients in \eqref{eq:U2psi} satisfy $c_1(t) + (N+1)c_2(t) = 1$. For a single many-body eigenstate this is the elementary statement that $\langle A(t)\rangle = \langle\Psi^{N_p}|A|\Psi^{N_p}\rangle$ is exactly stationary. This distinction is important: the stationarity of an individual many-body eigenstate is not in conflict with the relaxation found in Sec.~\ref{sec:thermalization}, because the latter concerns generic (non-eigenstate) initial states and their ensemble-dependent overlaps with the eigenbasis. Thus $\langle C(t)\rangle$ itself carries no thermalization dynamics; the relaxation of Sec.~\ref{sec:thermalization} is visible in the model only through the initial-state dependence and the fluctuations studied next. Any approximation that replaced the connected part of the second moment by its Gaussian value $\frac{1}{N}(\delta_{ik}\delta_{jl}+\delta_{il}\delta_{jk})$ would violate the exact trace relation $\operatorname{Tr}C(t)=N_p$, since $\operatorname{Tr}C(t)=N_p$ identically.

The covariance of $C_{ij}(t)$ can be obtained by similar techniques. A particularly important quantity is the variance of $\langle A(t) \rangle$ itself:
\begin{equation}
\operatorname{Var}\!\bigl( \langle A(t) \rangle \bigr) = \sum_{ijkl} A_{ij} A_{kl} \, \operatorname{Cov}\!\bigl( C_{ij}(t), C_{kl}(t) \bigr).
\end{equation}
At $t=0$ the result follows directly from \eqref{eq:covcorr}:
\begin{equation}
\operatorname{Var}\!\bigl( \langle \Psi^{N_p}| A |\Psi^{N_p}\rangle \bigr)
= \frac{2 N_p (N-N_p)}{(N-1)N(N+2)}
\left[ \operatorname{Tr}(A^2) - \frac{(\operatorname{Tr}A)^2}{N} \right].
\label{eq:var_static}
\end{equation}
 At finite times the thermalization factor $|g(t/\tau_\lambda)|^2$ dresses this expression, and one finds at leading order
\begin{equation}
\operatorname{Var}\!\bigl( \langle A(t) \rangle \bigr) = \frac{2 f(1-f)}{N+2} \left[ \operatorname{Tr}(A^2) - \frac{(\operatorname{Tr}A)^2}{N} \right] |g(t/\tau_\lambda)|^2 + O(1/N^2),
\label{eq:var_A}
\end{equation}
generalizing Mag\'an's static result \eqref{eq:covcorr} to finite times. The factor $|g(t/\tau_\lambda)|^2$ shows that the variance of the fluctuations decays together with the expectation value itself. For large $t$, both vanish as $t^{-3}$. We emphasize that \eqref{eq:var_A} is a leading-order statement: at $t=0$ it reproduces the exact static variance \eqref{eq:var_static} up to the factor $N/(N-1)=1+O(1/N)$, which is absorbed in the $O(1/N^2)$ remainder, and the remainder also absorbs the spectral-correlation contributions (the analogue of $\operatorname{Var}(|f(t)|^2)$), which are $O(1/N^2)$ at fixed scaled time and have been discussed in Sec.~\ref{sec:variance}.

In contrast, for a full GOE Hamiltonian without particle number conservation, the analogous variance would be $\frac{1}{N}\operatorname{Tr}(A^2) |g(t/\tau_\lambda)|^2$, independent of any filling fraction. Thus, the free fermion fluctuations are suppressed by an additional factor $f(1-f)$, which reaches its maximum $1/4$ at half‑filling. This reflects the reduced randomness due to the constraint of fixed particle number. 

\subsection{Single-particle spectral form factor}
\label{sec:SFF_comparison}

The separation identified above can be tested dynamically by changing the fine-grained spectral correlations while keeping the eigenvector statistics and one-point density fixed. In our free fermion model, the natural object that carries information about single-particle energy correlations is
\begin{equation}
S_1(t) \equiv |f(t)|^2 = \frac{1}{N^2} \sum_{\alpha,\beta} e^{-i(E_\alpha-E_\beta)t},
\end{equation}
where $f(t) = \frac{1}{N}\sum_\alpha e^{-iE_\alpha t}$ as defined in Eq.~\eqref{eq:U1}.  We will refer to $S_1(t)$ as the \textit{single-particle spectral form factor}.  Its ensemble average $\langle S_1(t)\rangle$ controls the leading thermalization dynamics through Eq.~\eqref{eq:main_avg}, and its $O(1/N)$ corrections encode the presence or absence of spectral chaos.

In the Poissonian ensemble (cf. Appendix \ref{PE}), the eigenvalues $E_\alpha$ are independent random variables drawn from the average density $\rho(E)$ of Eq.~\eqref{eq:rhoE}---the same semicircular one-point density as in the GOE ensemble, so that the comparison below isolates the effect of the fine-grained correlations.  For independent variables, the average of $S_1(t)$ is easily computed:
\begin{align}
\langle S_1(t) \rangle_{\mathrm{P}} &= \frac{1}{N^2} \sum_{\alpha,\beta} \langle e^{-iE_\alpha t} e^{iE_\beta t} \rangle \nonumber \\
&= \frac{1}{N^2} \left[ N \langle |e^{-iEt}|^2 \rangle + N(N-1) |\langle e^{-iEt}\rangle|^2 \right] \nonumber \\
&= |\langle f(t) \rangle|^2 + \frac{1}{N} \left(1 - |\langle f(t) \rangle|^2 \right),
\label{eq:S1_P}
\end{align}
where we used $\langle |e^{-iEt}|^2 \rangle = 1$ and $\langle e^{-iEt}\rangle = \int dE\,\rho(E) e^{-iEt} = e^{-i\alpha t/\hbar} g(t/\tau_\lambda)$ from Eq.~\eqref{eq:favg}. Equation \eqref{eq:S1_P} is an \emph{exact} finite-$N$ identity for the Poissonized ensemble (only the average density \eqref{eq:rhoE} enters).  The first term $|\langle f(t)\rangle|^2 = |g(t/\tau_\lambda)|^2$ is exactly the leading relaxation envelope.  The second term is a smooth correction of order $1/N$ that persists for all times. For large $N$ it is negligible compared to the leading term except when $|g|^2$ itself becomes $O(1/N)$.

For the original GOE/GUE Hamiltonian, the eigenvalues are not independent, and their joint distribution includes the Vandermonde determinant, which induces level repulsion and long-range correlations.  The ensemble average of $S_1(t)$ now receives an additional contribution from the connected two-point correlation function.  Using the decomposition $\langle |f(t)|^2 \rangle = |\langle f(t)\rangle|^2 + \langle |f(t)|^2 \rangle_{\mathrm{corr}}$, we have
\begin{equation}
\langle S_1(t) \rangle_{\mathrm{RMT}} = |g(t/\tau_\lambda)|^2 + \frac{1}{N} \mathcal{F}\!\left(\frac{t}{\tau_d}\right) + \frac{1}{N} \left(1 - |g(t/\tau_\lambda)|^2\right),
\label{eq:S1_RMT}
\end{equation}
which is a \emph{matched asymptotic expression} rather than an exact finite-$N$ identity: the exact initial condition $\langle S_1(0)\rangle=1$ is reproduced exactly, through the cancellation in \eqref{eq:F_def}; the long-time plateau $1/N$ is carried by the smooth term alone (the two-level term vanishes); and the formula is valid to order $1/N$ in the local-bulk approximation, the finite-$N$ density and edge corrections being of relative order $1/N$ in the envelope and $O(1/N^2)$ in the two-level term.
where $\tau_d = \hbar/d$, $\widetilde b(s)$ is the Fourier transform of the GOE two-level cluster function $Y_2(y)$ in the local (bulk) approximation described in Sec.~\ref{sec:eigenbasis_vs_spectral}, $\mathcal{F}(t/\tau_d)=|g(t/\tau_\lambda)|^2-\widetilde b(t/\tau_d)$ is the two-level contribution beyond the smooth part, as defined in \eqref{eq:F_def}, and the term $\frac{1}{N}(1-|g|^2)$ is the same smooth part as in the Poissonian case.

\begin{figure}[t]
\centering
\includegraphics[width=0.62\textwidth]{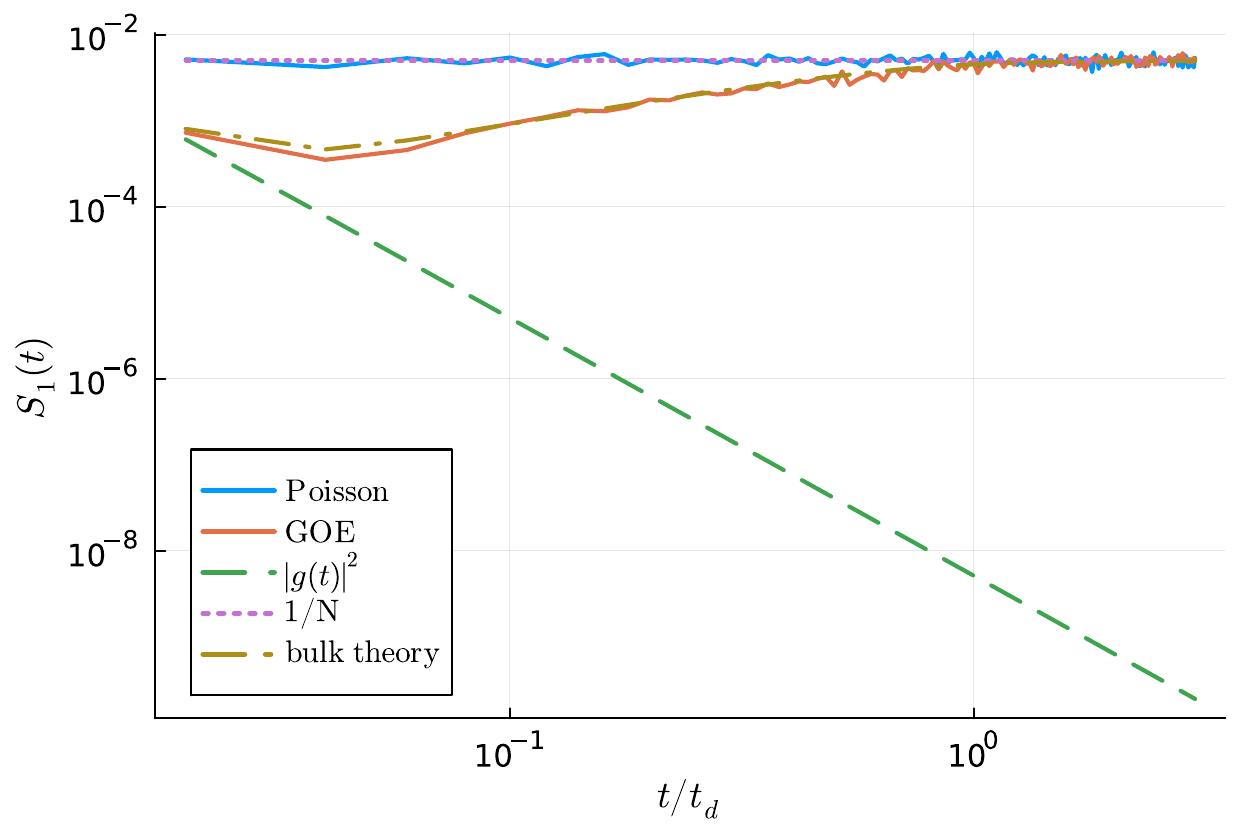}
\caption{Single-particle spectral form factor $\langle S_1(t)\rangle$ for the GOE (solid) and its Poissonized partner (dashed), $N=200$, averaged over $100$ realizations, plotted against $t/\tau_d$ on log-log scales. Both ensembles share the leading decay $|g(t)|^2$ (short-dashed) and the smooth plateau $1/N$ (dotted); the GOE curve dips below its Poissonian partner at intermediate times---the correlation hole encoded in $\mathcal{F}(t/\tau_d)$---and merges with it on the plateau. The dash-dotted ``bulk theory'' curve is the local-bulk prediction \eqref{eq:S1_RMT} (the GOE data are computed from the full finite-$N$ spectra without unfolding); it lies on the GOE data within the statistical resolution. This is the dynamical separation of eigenbasis chaos (common envelope) from spectral chaos (GOE-only two-level structure).}
\label{fig:sff}
\end{figure}

Comparing Eqs.~\eqref{eq:S1_P} and \eqref{eq:S1_RMT}, we see that the leading term $|g|^2$ is identical in both ensembles.  The difference lies entirely in the $O(1/N)$ corrections.  In the Poissonian ensemble, the correction is smooth and given by $N^{-1}(1-|g|^2)$.  In the random-matrix ensemble, there is an additional two-level component $N^{-1}\mathcal{F}(t/\tau_d)$ that reflects the eigenvalue correlations. Cf. Figure \ref{fig:sff}, where the numerical curves are computed from the full finite-$N$ spectra without unfolding, while $\mathcal{F}$ in \eqref{eq:S1_RMT} is the local-bulk form of Sec.~\ref{sec:eigenbasis_vs_spectral}: at the present system size, spectral average and time normalization, the full-spectrum finite-$N$ data exhibit a dip of the GOE curve below its Poissonian partner, and the excess of the GOE data over $|g|^2+N^{-1}(1-|g|^2)$ agrees with $N^{-1}\mathcal{F}(t/\tau_d)$ within the standard error of the $100$-realization average (about $7\times10^{-5}$ at $t/\tau_d=0.05$, rising to $5\times10^{-4}$ at $t/\tau_d\approx1$; the largest deviation over $0.1\le t/\tau_d\le 1.5$ is three standard errors). The dip is thus quantitatively consistent with the local-bulk correlation structure encoded in $\mathcal{F}$, with the remaining deviations of the expected size of the finite-$N$ density and edge effects.

\section{Conclusion}
\label{sec:conclusion}

We have asked which statistical structures of a random many-body Hamiltonian control the different layers of thermalization and eigenstate thermalization, and answered it in the random free fermion model, where the GOE one-body Hamiltonian has Haar-distributed eigenvectors statistically independent of its eigenvalues. What we learned is a four-way attribution. The Haar eigenvector statistics determine the ETH fluctuation scale \eqref{eq:eth_condvar_ens} and the off-diagonal amplitude \eqref{eq:eth_spectral_function}. The one-point semicircle density determines the leading relaxation envelope $g^2(t/\tau_\lambda)$ with $\tau_\lambda=\hbar/(2\eta\sqrt{N})$ and its $|t|^{-3}$ tail, while the variance of the relaxation fluctuations vanishes as $O(1/N)$, so that the envelope self-averages at fixed scaled time. The exact eigenvector-eigenvalue independence explains, through the factorization theorem \ref{thm:factorization}, why the ETH structure of this model is energy independent at the ensemble level; at fixed realization the energy-shell conditioning enters only through a covariance subtraction of relative order $1/N$, and for $k$-body observables the same suppression holds for every fixed $k$. A controlled correlation deformation $H(\varepsilon)=H_0+\varepsilon H_1$ then generates, at first order in $\varepsilon$, the energy dependence of both the diagonal energy-resolved slope and the off-diagonal ETH spectral function. Finally, two-level spectral correlations produce only the subleading $O(1/N)$ late-time fingerprint $\mathcal{F}(t/\tau_d)$ in the single-particle spectral form factor, without modifying the leading envelope. A methodological point of independent interest is that the $O(1/N)$ accuracy of the leading channels throughout requires the exact Haar (Weingarten) second moment of the eigenvector matrix: the Gaussian approximation would violate the exact initial condition at $t=0$ and misstate the coefficients of the eigenvector-fluctuation channel.

The model thereby provides a controlled separation of eigenbasis statistics, coarse-grained spectral density, eigenvector-eigenvalue correlations, and fine-grained spectral correlations: each ingredient of the random Hamiltonian is tied to a distinct physical layer of thermalization or ETH. The separation is exact in the undeformed GOE free-fermion ensemble; the deformation provides a controlled, computable mechanism for generating energy dependence---not a general theory of energy-dependent ETH in chaotic systems, and not a claim that this mechanism is unique. The energy dependence is generated by switching on the eigenvector--eigenvalue correlations that are absent in the undeformed ensemble.

Three directions are most natural. The first is to extend the deformation \eqref{eq:deform_H} to interacting perturbations and to higher orders in $\varepsilon$, connecting the mechanism exhibited here to generic many-body systems. The second is the converse experiment of replacing the Haar eigenbasis by a fixed (non-random) basis while keeping the GOE spectrum, which completes the two-by-two separation table of eigenbasis versus spectrum and thereby closes the comparison between eigenbasis chaos and spectral chaos in this solvable setting. A third direction is to identify the
effective energy window of the deformation, i.e.\ to match the
perturbatively generated eigenvector--eigenvalue correlations
\eqref{eq:deform_cov} to a block-Haar ensemble $U_{\mathrm{loc}}$ acting on
energy windows of width $\Delta_\varepsilon$, and to extend the exact
orthogonal-class coefficients obtained here to the higher-$k$ cactus
diagrams whose symmetry-class dependence was left open in
Ref.~\cite{Vallini2026}.

\section*{Acknowledgements} 
We would like to thank Elisa Vallini for helpful comments.
X.G. is supported by Yancheng Institute of Technology (xjr2024030).  Z.H. is supported by the National Natural Science Foundation of China under Grant No. 12305035.
\appendix
\section{Contraction rules for higher moments of $U(t)$}
\label{app:contractions}

In this appendix we list the explicit expressions for the connected parts of the second, third, and fourth moments of $U(t)$ after eigenvector averaging. These formulas are used in the variance analysis of Sec.~\ref{sec:variance}.

\paragraph{Exact fourth moment of the Haar measure.}
The eigenvectors $\psi_i^a$ form a Haar-distributed matrix in $O(N)$, whose exact Weingarten function yields
\begin{equation}
\langle \psi_i^a \psi_j^a \psi_k^b \psi_l^b \rangle
= \frac{\bigl(N+1-2\delta_{ab}\bigr)\delta_{ij}\delta_{kl} + \bigl(N\delta_{ab}-1\bigr)\bigl(\delta_{ik}\delta_{jl}+\delta_{il}\delta_{jk}\bigr)}{(N-1)N(N+2)}.
\label{eq:haar4}
\end{equation}
This  reproduces the exact second moment \eqref{eq:psicov} upon contraction, the exact same-column variance
 \begin{equation}\langle (\psi_i^a)^2(\psi_j^a)^2\rangle = \frac{1+2\delta_{ij}}{N(N+2)},\end{equation}
 the exact cross-column value 
 \begin{equation}\langle (\psi_i^a)^2(\psi_j^b)^2\rangle = \frac{N+1-2\delta_{ij}}{(N-1)N(N+2)}\end{equation}
 for $a\neq b$ (note that the Gaussian approximation would replace the cross-column value by $1/[N(N+2)]$, missing the relative $O(1/N)$ correction), and the exact unitarity constraints $\sum_i \psi_i^a \psi_i^b = \delta_{ab}$. Summing \eqref{eq:haar4} against $e^{-i(E_a-E_b)t/\hbar}$ over $a,b$ gives the exact second moment \eqref{eq:U2psi}. The same contraction summed against $e^{-i(E_a+E_b)t/\hbar}$ gives the exact \emph{unconjugated} pairing used in the eigenvector-fluctuation channel of Sec.~\ref{sec:variance}:
\begin{equation}
\langle U_{ik}(t)\, U_{mo}(t) \rangle_{\psi,\mathrm{corr}}
= \frac{2\bigl(f(t)^2-f(2t)\bigr)\,\delta_{ik}\delta_{mo} - N\bigl(f(t)^2-f(2t)\bigr)\bigl(\delta_{im}\delta_{ko}+\delta_{io}\delta_{km}\bigr)}{(N-1)(N+2)}.
\label{eq:UUncorr}
\end{equation}
It vanishes at $t=0$ (where $f(0)=f(0\cdot 2)=1$), in contrast to the Gaussian value $\frac{1}{N}\bigl(\delta_{ik}\delta_{jl}+\delta_{il}\delta_{jk}\bigr)f(2t)$ quoted below, another manifestation of the exact initial condition.

\paragraph{Gaussian (Wick) cumulants.}
The third and fourth cumulants below are obtained by applying Wick's theorem to the Gaussian variables $\psi_i^a$ and using \eqref{eq:psicov}. They are valid to leading order in $1/N$ only. The exact Haar cumulants differ from them by corrections of relative order $1/N$, i.e.\ of the same parametric size as the cumulants themselves, and the coefficients of the $O(1/N^2)$ and $O(1/N^3)$ variance contributions of Sec.~\ref{sec:variance} therefore inherit such corrections. The results parallel those derived by Weidenm\"uller for the GOE time‑evolution operator \cite{Weidenmuller2024a}, with the only difference that here the eigenvalues $E_a$ are random and will be averaged later.
\begin{align}
\langle U_{ij}(t) U_{kl}(t) \rangle_{\psi,\mathrm{corr}} &= \frac{1}{N}\bigl( \delta_{ik}\delta_{jl} + \delta_{il}\delta_{jk} \bigr) f(2t), \\
\langle U_{ij}(t) U_{kl}(t) U_{mn}^*(t) \rangle_{\psi,\mathrm{corr}} &= \frac{1}{N^2} \bigl[ \delta_{jk}\delta_{lm}\delta_{ni} + \text{7 permutations} \bigr] f(t), \\
\langle U_{ij}(t) U_{kl}(t) U_{mn}^*(t) U_{pq}^*(t) \rangle_{\psi,\mathrm{corr}} &= \frac{1}{N^3} \bigl[ \delta_{jk}\delta_{lm}\delta_{np}\delta_{qi} + \text{permutations} \bigr].
\end{align}
The factors $f(2t)$ and $f(t)$ arise when contractions force the eigenstate indices to be equal, producing a sum over $e^{-i2E_a t/\hbar}$ or $e^{-iE_a t/\hbar}$, respectively. When averaging over eigenvalues, these become $\langle f(2t)\rangle$ and $\langle f(t)\rangle$, which are $O(1)$ and decay on the scale $\tau_\lambda$. Note that the second moment is \emph{not} included in the Wick list above: for it we use the exact expression \eqref{eq:U2psi}, whose connected part \eqref{eq:U2corr} vanishes at $t=0$ and carries the essential factor $\bigl(1-|f(t)|^2\bigr)$.

\section{Poissonian ensemble}
\label{PE}
A powerful way to demonstrate that the leading relaxation dynamics is insensitive to fine-grained spectral statistics is to compare our random-matrix model with its associated \textit{Poissonian ensemble}, following the ideas of Mag\'an and Wu \cite{MaganWu2024}.  

Given a Hamiltonian $H$ with eigenvectors $\{|\alpha\rangle\}$ and eigenvalues $\{E_\alpha\}$ distributed according to some joint probability law, the Poissonian ensemble is defined by keeping the eigenvectors unchanged but replacing the eigenvalues by a set of independent random variables drawn from the same average density of states $\rho(E)$.  In the context of our free fermion model, where the eigenvectors are Haar-distributed, this construction yields an ensemble of Hamiltonians $H^{\text{(P)}}$ with the same eigenvector statistics as the original GOE/GUE model, but with eigenvalues that are independent and identically distributed (i.i.d.) according to $\rho(E)$ of Eq.~\eqref{eq:rhoE}.  

The order of the eigenvalues is irrelevant for the spectral statistics considered here: all quantities below (linear statistics of the eigenvalue set, such as $f(t)$ and $S_1(t)$) are invariant under relabelling of the eigenvalues. Sorting is introduced only to fix a labelling convention for the eigenvector-eigenvalue pairing, and does not affect any of the averages computed in this appendix.

\paragraph{Average of $f(t)$ in the Poissonian ensemble.}
The function $f(t) = \frac{1}{N}\sum_\alpha e^{-iE_\alpha t}$ is a linear statistic of the eigenvalues (in this section we set $\hbar=1$).  In the Poissonian ensemble, the eigenvalues are i.i.d., so its average is simply
\begin{equation}
\langle f(t) \rangle_{\text{P}} = \int dE\, \rho(E)\, e^{-iEt} = \langle f(t) \rangle_{\text{RMT}},
\label{eq:f_avg_P}
\end{equation}
where $\langle \cdot \rangle_{\text{RMT}}$ denotes the average in the original random-matrix model.  Thus, the leading relaxation envelope $g(t/\tau_\lambda) = e^{i\alpha t/\hbar}\langle f(t)\rangle$ is identical in both ensembles.  This already indicates that the leading relaxation, governed by the Fourier transform of the average level density, is unaffected by the presence or absence of spectral correlations.

\paragraph{Second moment of $f(t)$ in the Poissonian ensemble}

For i.i.d. eigenvalues, the second moment is straightforward to compute:
\begin{align}
\langle |f(t)|^2 \rangle_{\text{P}} &= \frac{1}{N^2} \sum_{\alpha,\beta} \langle e^{-i(E_\alpha-E_\beta)t} \rangle \nonumber \\
&= \frac{1}{N^2} \left[ N \langle |e^{-iEt}|^2 \rangle + N(N-1) |\langle e^{-iEt}\rangle|^2 \right] \nonumber \\
&= |\langle f(t) \rangle|^2 + \frac{1}{N} \left(1 - |\langle f(t) \rangle|^2 \right).
\label{eq:f2_P}
\end{align}
Here we used that $\langle |e^{-iEt}|^2 \rangle = 1$ because $|e^{-iEt}|=1$.  The correction term is smooth in time and of order $1/N$.  For large $N$, it is negligible compared to the leading term $|\langle f(t)\rangle|^2$ at early and intermediate times, but it becomes important when $|\langle f(t)\rangle|^2$ decays to $O(1/N)$.

In contrast, for the original random-matrix model (GOE/GUE), the second moment receives an additional contribution from eigenvalue correlations \cite{Weidenmuller2024a}:
\begin{equation}
\langle |f(t)|^2 \rangle_{\text{RMT}} = |\langle f(t) \rangle|^2 + \frac{1}{N} \mathcal{F}\!\left(\frac{t}{\tau_d}\right) + \frac{1}{N}\left(1 - |\langle f(t) \rangle|^2\right).
\label{eq:f2_RMT}
\end{equation}
Here $\tau_d = \hbar/d$ is the Heisenberg time (with $d$ the mean level spacing), $\widetilde b(s)$ is the Fourier transform of the GOE two-level cluster function $Y_2(y)$ in the local (bulk) approximation of Sec.~\ref{sec:eigenbasis_vs_spectral} (normalized by $\widetilde b(0)=1$), and $\mathcal{F}(t/\tau_d)=|g(t/\tau_\lambda)|^2-\widetilde b(t/\tau_d)$, as defined in \eqref{eq:F_def}.  Like \eqref{eq:S1_RMT}, Eq.~\eqref{eq:f2_RMT} is a matched asymptotic expression, exact at $t=0$ through the cancellation in \eqref{eq:F_def} and valid to order $1/N$ in the local-bulk approximation.  The term $\frac{1}{N}(1 - |\langle f(t)\rangle|^2)$ is the same smooth part as in the Poissonian case, while $\frac{1}{N}\mathcal{F}(t/\tau_d)$ is the two-level signature of spectral chaos (the correlation hole).

\bibliographystyle{amsalpha}

\end{document}